\documentclass[10pt]{article}
\usepackage[margin=0.6in]{geometry}
\usepackage{hyperref}
\usepackage{multicol}
\usepackage{blindtext}
\usepackage{amsmath}
\usepackage{amssymb}
\usepackage{mathrsfs}
\usepackage{dsfont}
\usepackage{framed}
\usepackage{algorithm}
\usepackage{algpseudocode}
\usepackage{authblk}
\usepackage{graphicx}
\usepackage{wrapfig}
\usepackage[table]{xcolor}
\usepackage{wrapfig}
\usepackage{float}
\usepackage{mdframed}
\usepackage{multirow}
\mdfdefinestyle{mddefault}{
rightline=true,innerleftmargin=4pt,innerrightmargin=4pt,innertopmargin=4pt,innerbottommargin=2pt
frametitlerule=false,backgroundcolor=black!5}

\usepackage{framed,color}
\definecolor{shadecolor}{rgb}{0.85,0.85,0.85}
\usepackage{ulem}
\usepackage{url}
\usepackage{cleveref}
\usepackage{hyperref,xcolor}
\hypersetup{
colorlinks,
linkcolor={red!50!black},
citecolor={blue!50!black},
urlcolor={blue!80!black}
}
\usepackage{xr}
\makeatletter
\newcommand*{\addFileDependency}[1]{% argument=file name and extension
\typeout{(#1)}% latexmk will find this if $recorder=0
\@addtofilelist{#1}
\IfFileExists{#1}{}{\typeout{No file #1.}}
}\makeatother

\usepackage{color}
\usepackage{ulem}
\usepackage{cancel}

\usepackage{lineno}
\usepackage{amsthm}
\usepackage{hyperref}

\theoremstyle{definition}
\newtheorem{Theorem}{Theorem}
\newtheorem{Lemma}{Lemma}
\newtheorem{Corollary}{Corollary}[Theorem]

\RequirePackage{tikz} % For \foreach used for Orcid icon
\newcommand{\orcidicon}{\includegraphics[width=0.32cm]{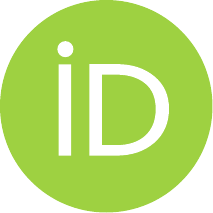}}

\foreach \x in {A, ..., Z}{%
\expandafter\xdef\csname orcid\x\endcsname{\noexpand\href{https://orcid.org/\csname orcidauthor\x\endcsname}{\noexpand\orcidicon}}
}

\begin{document}
\title{A Theoretical Review of Area Production Rates as Test Statistics for Detecting Nonequilibrium Dynamics in Ornstein-Uhlenbeck Processes}
\newcommand{\orcidauthorA}{0000-0001-6618-631X} % Add \orcidA{} behind the author's name

\author{Alexander Strang $^1$  \orcidA{} }

\affil{\footnotesize $^1$ University of California, Berkeley, Statistics. \\ Correspondence: alexstrang@berkeley.edu}

\date{\vspace{-1.2em} \today}

\maketitle
\vspace{-20pt}

%%%%%%%%%%%%%%%%%%%%%%%%%%%%%%%%%%%%%%%%%%%%%%%%%%%%%%
\section{Abstract}

A stochastic process is at thermodynamic equilibrium if it obeys time-reversal symmetry; forward and reverse time are statistically indistinguishable at steady state. Non-equilibrium processes break time-reversal symmetry by maintaining circulating probability currents. In physical processes, these currents require a continual use and exchange of energy. Accordingly, signatures of non-equilibrium behavior are important markers of energy use in biophysical systems. In this article we consider a particular signature of nonequilibrium behavior: area production rates. These are, the average rate at which a stochastic process traces out signed area in its projections onto coordinate planes. Area production is an example of a linear observable: a path integral over an observed trajectory against a linear vector field. We provide a summary review of area production rates in Ornstein-Uhlenbeck (OU) processes. Then, we show that, given an OU process, a weighted Frobenius norm of the area production rate matrix is the optimal test statistic for detecting nonequilibrium behavior in the sense that its coefficient of variation decays faster in the length of time observed than the coefficient of variation of any other linear observable. We conclude by showing that this test statistic estimates the entropy production rate of the process.

%%%%%%%%%%%%%%%%%%%%%%%%%%%%%%%%%%%%%%%%%%
\section{Introduction and Background} \label{sec: intro}

% nonequilibrium processes are important
Many important physical processes operate far from thermodynamic equilibrium \cite{parrondo2009entropy}. Examples range from essential cellular processes occurring at the molecular scale \cite{gnesotto2018broken} to oscillations in the climate \cite{weiss2019nonequilibrium}. Unlike equilibrium processes, whose thermodynamics are well understood \cite{mura2019mesoscopic}, nonequilibrium processes are harder to study, and are largely characterized by expansions near equilibrium, or by fluctuation theorems \cite{jiang2004mathematical,gingrich2016dissipation}. Equilibrium processes satisfy time-reversal symmetry; when initialized from the steady-state distribution, it is impossible to distinguish trajectories evolved forward in time from trajectories evolved backwards in time. Nonequilibrium processes break time-reversal symmetry by maintaining nonvanishing probability currents at steady state \cite{tomita1974irreversible}. These cyclic currents manifest in hysteretic fluctuations and noisy oscillation. Such cycles can only be maintained by a continuous exchange of energy between two reservoirs via the system of interest \cite{li2019quantifying}. In contrast, an energetically isolated system, or system in contact with a single reservoir, cannot maintain steady state probability currents \cite{teitsworth2022stochastic}. This property of equilibrium processes is called detailed balance. Violations of detailed balance signify that a process is not at equilibrium.

%% example with an OU process, commutator condition
For example, consider an Ornstein-Uhlenbeck (OU) process \cite{uhlenbeck1930theory} evolving in $\mathbb{R}^d$. Let $X(t) \in \mathbb{R}^d$ denote the state of the process at time $t$. Then, $X(t)$ is a Markov process that obeys the stochastic differential equation (SDE):
\begin{equation} \label{eqn: OU process SDE}
    dX(t) = \mu(X(t)) dt + G dW(t), \quad \mu(x) = - A x
\end{equation}
where $A$ and $G$ are both matrices in $\mathbb{R}^{d \times d}$, $dW(t)$ denotes a Weiner increment (white noise with variance $dt$), and where $W(t)$ denotes the Weiner process in $\mathbb{R}^d$ (Brownian motion) \cite{ahmad1988introduction,gardiner1985handbook}. We interpret the SDE \eqref{eqn: OU process SDE} in the It\^o sense \cite{van1981ito}. We will, throughout this manuscript, assume that $A$ is full rank, diagonalizable, and the real part of all of the eigenvalues of $A$ are positive. We will also assume that $G$ is full rank.

Equation 
\eqref{eqn: OU process SDE} is the simplest stochastic analog to an asymptotically stable linear system of ordinary differential equations (ODE's). Accordingly, OU processes are widely used to approximate the behavior of more general SDE's near stable equilibria of the corresponding ODE, $\frac{d}{dt} x(t) = \mu(x(t))$ \cite{patterson2021and}. Alternately, an OU process of the form equation \eqref{eqn: OU process SDE}, may be used to generate colored noise traces \cite{bibbona2008ornstein}, to study stochastic oscillators \cite{gonzalez2019experimental,teitsworth2022stochastic}. All linear autoregressive models with Gaussian noise are time-sampled OU processes. Generically, solutions to \eqref{eqn: OU process SDE} are samples from a vector-valued Gaussian process with an exponential kernel \cite{williams2006gaussian}. These may be represented by convolving the noise process $G dW(t)$ with the impulse response function $f(t) = \{0 \text{ if } t < 0, \exp(-A t) \text{ if } t \geq 0\}$.  

Suppose that, at $t = 0$, $X(0) \sim \pi(\cdot,0)$. Then, let $\pi(\cdot,t)$ denote the probability density function for $X(t)$. The density function obeys the Fokker-Planck equation \cite{fokker1914mittlere,planck1917satz,van1992stochastic}:
\begin{equation} \label{eqn: OU Fokker-Planck}
    \partial_t \pi(x,t) = - \nabla \cdot \mu(x) \pi(x,t) + \frac{1}{2} D \nabla^2 \pi(x,t) = -\nabla \cdot j(x,t)
\end{equation}
where $D = G G^{\intercal}$ is the diffusion tensor, and where $j(x,t)$ denotes the flux in probability density at state $x$ and time $t$ \cite{gardiner1985handbook}:
\begin{equation} \label{eqn: density flux}
    j(x,t) = \mu(x) \pi(x,t) - \frac{1}{2} D \nabla \pi(x,t) = -\left(A x \pi(x,t) + \frac{1}{2} D \nabla \pi(x,t) \right).
\end{equation}

A distribution $\pi_*(\cdot)$ is a steady state for the Fokker-Planck dynamic if the corresponding fluxes, $j_*(x)$ have zero divergence. In general, this can be accomplished by any $\pi_*$ such that $j_*(x)$ is incompressible. This is a flux balance equation. It requires that the total flux into any region equals the total flux out of that region. It does \textit{not} require that the flux itself is zero. If the flux $j_*(\cdot)$ vanishes at steady state, i.e.~$j_*(x) = 0$ for all $x$, then the process obeys \textit{detailed} balance. There is no circulation of probability density at equilibrium.

Solutions to \eqref{eqn: OU Fokker-Planck} converge to a Gaussian steady-state $\pi_*(\cdot)$ with steady-state covariance $\Sigma_*$. The covariance is the unique solution to the Lyapunov equation \cite{tomita1974irreversible,van1992stochastic}:
\begin{equation} \label{eqn: Lyapunov}
    A \Sigma_* + \Sigma_* A^{\intercal} = D.
\end{equation}

Any Gaussian density with covariance $\Sigma_*$ satisfies $\nabla \pi_*(x) = \left(-\Sigma_*^{-1} x \right) \pi_*(x)$. Therefore, the steady-state fluxes can be expanded:
\begin{equation} \label{eqn: steady state velocities}
j_*(x) = v_*(x) \pi_*(x) \text{ where } v_*(x) = \left(\frac{1}{2} D \Sigma_*^{-1} - A  \right) x. 
\end{equation}
Here $v_*(x)$ is the steady-state velocity of the probability density at $x$ \cite{tomita1974irreversible,tomita2008irreversible}. The steady-state density is nonzero everywhere, so the steady-state fluxes are zero everywhere if and only if the steady-state velocities are zero everywhere. So, we can identify whether or not an OU process obeys detailed balance by studying $v_*(x)$. The left panel of Figure \ref{fig: velocities and area production} shows $v_*$ for a two-dimensional OU process. 

Applying the Lyapunov equation \eqref{eqn: Lyapunov} to $D$ in equation \eqref{eqn: steady state velocities} gives:
\begin{equation} \label{eqn: steady state velocities using alpha}
v_*(x) = \left(\frac{1}{2} D - A \Sigma_* \right) \Sigma_*^{-1} x = \frac{1}{2}\left(\Sigma_* A^{\intercal} - A \Sigma_* \right) \Sigma_*^{-1} x = - \alpha_* \Sigma_*^{-1} x
\end{equation}
where:
\begin{equation} \label{eqn: commutator for area production}
\alpha_* = \frac{1}{2} \left(A \Sigma_* - \Sigma_* A^{\intercal} \right)
\end{equation}

The matrix $\alpha_*$ plays a fundamental role in the analysis of nonequilibrium OU processes \cite{tomita1974irreversible,tomita2008irreversible}. In the following, we will show that $\alpha_*$ can be estimated by tracking the area traced out by long trajectories, is intimately related to the behavior of the steady state distribution when the degree of driving rotation inducing non-equilibrium circulation varies, and controls the production rate of a broad class of observables. 

Solutions to the Lyapunov equation are full-rank, so $v_*(x) = 0$ for all $x \in \mathbb{R}^d$ if and only if $\alpha_* = 0$. Therefore, an OU process obeys detailed balance if and only if $\alpha_* = 0$. We will see that $\alpha_*$ is zero if and only if $D^{-1} A$ is a symmetric matrix. This is a special case of the more general requirement that, an SDE obeys detailed balance if and only if $D^{-1}(x) \mu(x)$ can be expressed as the gradient of some potential function on $\mathbb{R}^d$ \cite{strang2020applications}. Clearly, generic $(A,G)$ will not satisfy this symmetry requirement, so almost all OU processes do not satisfy detailed balance. This situation is generic. Most parameterized SDE's only satisfy detailed balance on a low-dimensional manifold characterized by a set of symmetry constraints. In this sense, non-equilibrium processes are the mathematically general case.

% not always obvious if the process is nonequilibrium
While non-equilibrium processes are ubiquitous, it is not always obvious whether observed fluctuations are thermal fluctuations in an equilibrium process, or oscillatory dynamics arising from circulating probability currents \cite{skinner2021estimating}. For example, while many processes in cellular biology are clearly violate equilibrium (i.e. cell division, transcription and translation, sensing, biochemical patterning, or mechanical force generation) \cite{cao2015free,england2013statistical,nirody2017biophysicist,parrondo2009entropy,wang2010robust}, for others (i.e. dynamics of chromosomes, the cytoskeleton, the nucleus, or beating of flagella) it is not clear, and violations of detailed balance must be empirically demonstrated and quantified \cite{gladrow2017nonequilibrium,gnesotto2018broken,guo2014probing,mura2019mesoscopic,skinner2021estimating}. At the mesoscopic scale, large collections of non-equilibrium processes may have emergent dynamics which do not distinctly violate, or even regain, detailed balance \cite{egolf2000equilibrium}. Consequently, a variety of experimental procedures have been developed to detect violations of detailed balance.

%% experimental metrics overview
% some experimental methods exist but invasive
In fluctuation-dissipation based methods the study system is slightly perturbed by an applied force, the responses are measured, and compared to the power spectrum of thermal fluctuations in the process. While widely used, this experimental technique is invasive, since it requires perturbing the study system. Some microscopic systems are too fragile to perturb, while other systems are too large or complicated (cf.~\cite{weiss2019nonequilibrium}). In those cases it is important to have observational metrics which depend only on observed trajectories. 

% observational methods
Detecting time reversal asymmetry from observed trajectories is a classical problem in time series analysis \cite{steinberg1986time}. The conceptually simplest method is to use observed trajectories to reconstruct the probability velocity field \cite{gnesotto2018broken}. This technique is data intensive and requires long-trajectories, especially in high dimension, since it requires the process to repeatedly visit a series of voxels many times \cite{frishman2020learning,gonzalez2019experimental,li2019quantifying}. Reconstruction of velocity fields is accordingly limited to low-dimensional systems that are inexpensive to observe. When the underlying velocity field cannot be reconstructed its influence on trajectories can still be measured. For example, large fluctuations typically form hysteritic loops, that can be reconstructed from repeated observations. This technique is also information intensive since it requires the repeated observation of the same rare fluctuations \cite{gonzalez2019experimental}. Both of these techniques are expensive since they require reconstruction of local properties of the velocity field.

% motivates global metrics - stochastic area proposed by ghanta - actually related to tomita and tomita
These difficulties motivate the development of observational metrics that use the entire observed trajectory to infer a global property of the velocity field (c.f.~\cite{gradziuk2019scaling,teitsworth2022stochastic,weiss2019nonequilibrium}). Area production rates are promising global metrics that can be applied to observed trajectories, without requiring a dynamical model, or complete observation of all dynamical variables \cite{ghanta2017fluctuation,martinez2019inferring,teza2020exact}.
The area production rates of a process form a matrix, whose $i,j$ entry is the expected rate at which the projection of the process onto coordinates $i$ and $j$ traces out signed area \cite{ghanta2017fluctuation}. These rates can be estimated directly from observed trajectories and only require observation of the $i^{th}$ and $j^{th}$ coordinates \cite{frishman2020learning}. 

\begin{figure}[t]
\centering
\includegraphics[trim = 85 20 60 0, clip,scale = 0.45]{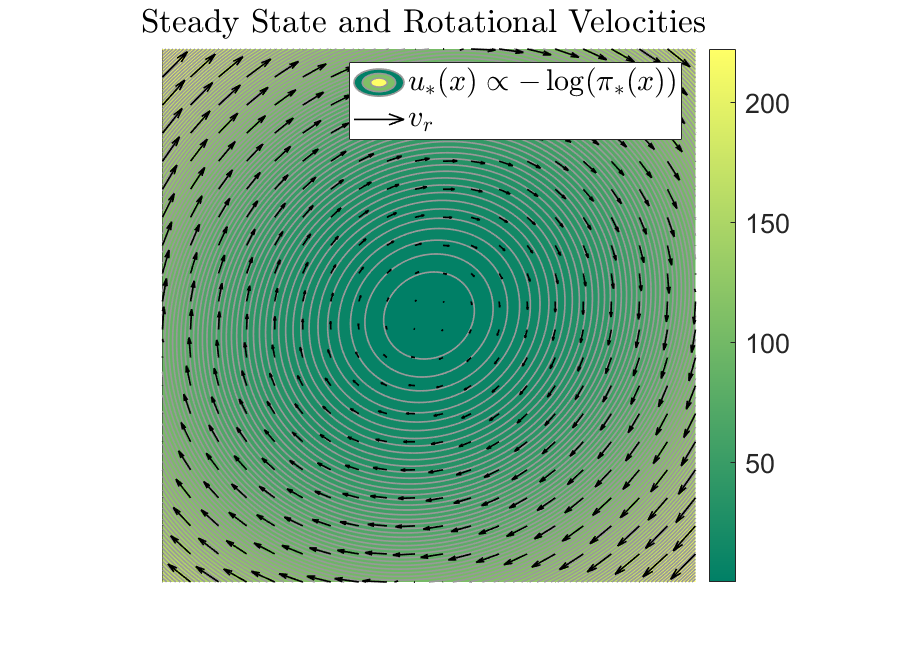}
\includegraphics[trim = 85 20 80 0, clip,scale = 0.45]{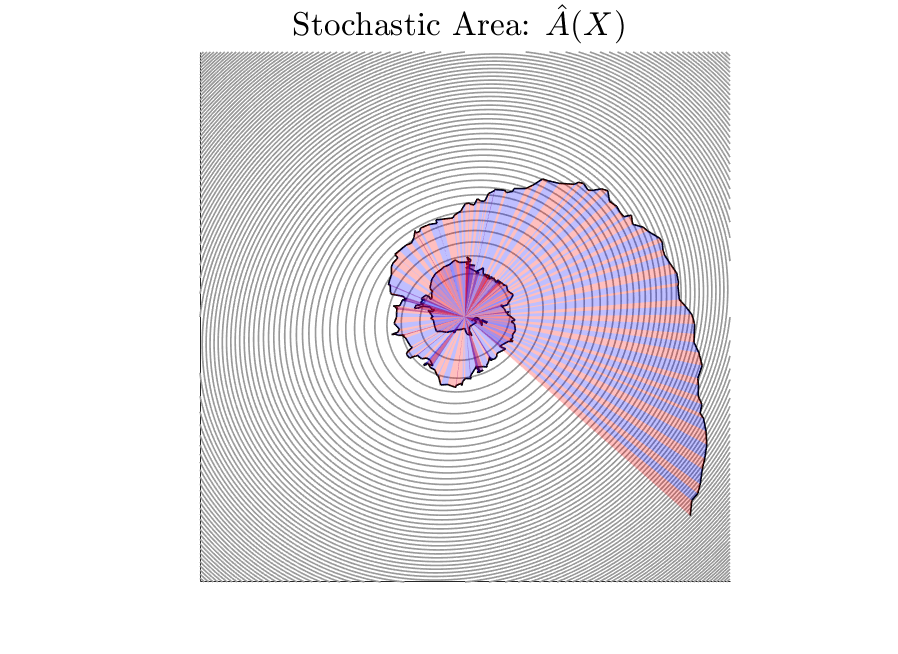}
\caption{\textbf{Left:} The negative log steady state, $-\log(\pi_*(x))$ for a two-dimensional OU process and the corresponding steady state velocities $v_*(x) = v_r(x)$. \textbf{Right:} An example trajectory  $\{X(t_k)\}_{k=0}^n$ sampled from the OU process and the associated signed area $\hat{A}(\{X(t_k)\}_{k=0}^n)$. The shaded triangles show the area produced by each step of the process. The total area produced is the sum of the signed area of each shaded triangle.}
\label{fig: velocities and area production}
\end{figure}   

Specifically, suppose that we observe a process $X$ at a sequence of times $\{t_k\}_{k=0}^n$ where $\Delta t = t_{k+1} - t_k$ is small, and $n \Delta t = T$. Then, we can approximate the area traced out by the trajectory in the $i,j$ plane via the path-sum:
\begin{equation}
\begin{aligned}
    \hat{A}_{ij}(\{X(t_k)\}_{k=0}^n)&  = \frac{1}{2} \sum_{j=0}^{n-1} [X_i(t_k),X_j(t_k)] \times [\Delta X_i(t_k),\Delta X_j(t_k)] \\
    & \hspace{1.5 cm} \approx \frac{1}{2} \int_{t = 0}^{T} [X_i(t),X_j(t)] \times [dX_i(t),dX_j(t)] = \hat{A}_{ij}(\{X(t)\}_{t=0}^T) \\
    \end{aligned}
\end{equation}
where $\times$ denotes the cross product, $\Delta X(t_k) = X(t_{k+1}) - X(t_{k})$, and where the approximation holds for small $\Delta t$. The right panel of Figure \ref{fig: velocities and area production} shows an example trajectory and the associated area.

Let $\hat{A}(\{X(t_k)\}_{k=0}^n)$ denote the matrix with entries $\hat{A}_{ij}(\{X(t_k)\}_{k=0}^n)$. Then, the average rate at which the process produced area on each coordinate plane is \cite{gonzalez2019experimental}:
\begin{equation}
\hat{\alpha}(\{X(t_k)\}_{k=0}^n) = \frac{1}{n \Delta t} \hat{A}(\{X(t_k)\}_{k=0}^n)) \approx \frac{1}{T} \hat{A}(\{X(t)\}_{t=0}^T) = \hat{\alpha}(\{X(t)\}_{t=0}^T).
\end{equation}

If the process is ergodic, then long-time averages converge in probability to averages against the steady-state distribution \cite{boltzmann1910vorlesungen, feller1991introduction}. Therefore, for large $T$, $\hat{\alpha}(\{X(t)\}_{t=0}^T)$ converges in probability to its expectation \cite{ghanta2017fluctuation}:
\begin{equation}
\lim_{T \rightarrow \infty} \hat{\alpha}_{ij}(\{X(t)\}_{t=0}^T) \longrightarrow \mathbb{E}_{X \sim \pi_*}\left[ \mathbb{E}_{dX|X}\left[ [X_i,X_j] \times [dX_i,dX_j] \right] \right] = \alpha_{ij}.
\end{equation}

The matrix $\alpha$ is the area production rate matrix for the process. Its $i,j$ entry returns the expected rate at which the process traces out area in the $i,j$ coordinate plane when initialized from its steady state. When ergodic, the empirical estimator $\hat{\alpha}$ along long sample trajectories converges, in expectation, to the steady-state area production rate matrix $\alpha$. We will show that the area production rate matrix $\alpha$ is equivalent to the angular momentum tensor of the probability density flux, and, given an OU process, converges to the commutator $\alpha_*$ defined in equation \eqref{eqn: commutator for area production}.

%% propose as a test
If a process obeys detailed balance it is equally likely to complete any circuit clockwise as counter-clockwise. Therefore, if a process obeys detailed balance, all area production rates must equal zero. Then, any nonzero rate indicates that a process does not obey detailed balance. We proved this condition for OU processes (see equations \eqref{eqn: steady state velocities} - \eqref{eqn: commutator for area production}). This test has been used to demonstrate violations of detailed balance in soft-matter systems \cite{gnesotto2020learning,mura2019mesoscopic} coupled RC circuits driven by separate noise sources \cite{gonzalez2019experimental}, and to explain oscillatory climate dynamics \cite{weiss2019nonequilibrium}.

% relation to physical properties 
The area production rate matrix is closely related to fundamental physical properties of the system and can reveal intrinsic features of the active driving forces \cite{mura2019mesoscopic}. It equals the angular momentum tensor of the steady state probability currents \cite{weiss2019nonequilibrium}, and is equivalent to the ``irreversible circulation of fluctuations" studied by Tomita and Tomita \cite{tomita2008irreversible,tomita1974irreversible}. The area production rate of a two-dimensional process is also proportional to the rate at which the process winds around the origin and is a better empirical measure since it does not become singular when trajectories approach the origin. If the underlying process is an OU process then inner products with the area production rate matrix control the rate of production of observables, including the entropy production of the process \cite{chiang2017electrical,gonzalez2019experimental,qian2002thermodynamics,tomita2008irreversible}. The area production rates also entirely specify the steady state probability currents and control the tilt and eccentricity of the isoclines of the steady state distribution relative to the equilibrium distribution. Recent work has suggested that the area production rate matrix can also be used to guide dimension reduction methods that retain as much information as possible about non-equilibrium dynamics \cite{gnesotto2020learning}.

% need to understand its empirical distribution for proper estimation/use as a test statistic
While estimation of area production rates from an observed trajectory is straightforward (cf.~\cite{gonzalez2019experimental}), the resulting estimate is itself a random variable since it depends on a sample from a stochastic process. The distribution of empirical area production rate will depend on where the trajectory is initialized, the underlying system dynamics, the sampling rate at which the process is observed,and the length of the observed trajectory. Understanding this distribution is required for proper statistical testing. It is important to note that sample trajectories of an equilibrium process will almost surely produce area, even though no area is produced in expectation. Thus a nonzero observed rate is not, taken alone, sufficient evidence to show that a process violates detailed balance. Instead, the observed rates must be far from zero relative to the uncertainty in the rates. To show that an observed production rate is statistically significant we must show that we would be unlikely to sample an area production rate equal to or larger than the observed rate if the process was at equilibrium. Some authors approximate the sampling distribution of the area production rate, and associated $p$-values, via bootstrapping procedures (c.f.~\cite{frishman2020learning,weiss2019nonequilibrium}). %Moreover, since the area production rates are entries of a matrix, and all of the entries must be zero for an equilibrium process, multiple testing techniques will have to be carefully applied when multiple area production rates are measured. For sufficiently simple models, including OU models which are widely used in the cited literature, it should be possible to derive the distribution analytically, or at least its asymptotics. Alternatively, since the area production rate can be computed independent of an assumed dynamical model, it is also important to develop bootstrapping approaches that can estimate the uncertainty in the observed area production from only the information present in observed trajectories (cf.~\cite{weiss2019nonequilibrium} or \cite{frishman2020learning}). These two approaches are complimentary, since bootstrapping methods may be improved if the distribution of observed area production rate has generic properties that can be revealed by analysis (i.e. if the observed area production rates obey a fluctuation theorem or are constrained by information theoretic considerations ) \cite{frishman2020learning}. 

%% motivating question: What is the best linear observable for OU processes? How to handle multiple testing/collapse to a single observable?
This raises an important question: is the area production rate a ``good" test statistic for detecting violations of detailed balance? In particular, does the uncertainty in the empirical area production rates decay quickly as the length of the observed trajectory increases? Moreover, since the area production rate is matrix-valued, and any nonzero entry indicates a violation of detailed balance, what multiple testing procedure, or what norm of the matrix, should be used? 

In this manuscript we will argue that area production rates are a good choice of test statistic for detecting violations of detailed balance in OU processes. We make this argument for two reasons:

\begin{enumerate}
%% two ways to pose this: 1. Define a class of observables and show optimality with respect to test uncertainty
%% two ways to pose this: 2. select a test statistic that is intrinsic to the problem (deep connections to other characteristics of the process, simple dependency on other measures of nonequilibrium behavior, should characterize the process when far from equilibrium)
\item \textbf{Interpretability:} If estimated, the area production matrix answers a variety of important questions regarding nonequilibrium behavior simply. We will show that the area production rate matrix characterizes the circulatory component of probability fluxes and velocities, is equivalent to the angular momentum of the probability fluxes, and is, in two dimensions, equal to a ratio of the rotational and conservative components of the vector field that controls maximally likely paths and the work exerted along a trajectory. Moreover, it controls the production rate of all observables defined as a path integral. Importantly this includes other signatures of nonequilibrium dynamics such as entropy production \cite{qian2002thermodynamics}. We will then focus on the two-dimensional case, and show that the area production rate provides an interpretable geometric relationship between the steady-state distribution of an OU process and the degree of underlying circulation. Many of these results are available in past work, but are not unified in a single review. We hope, by reviewing these properties, to provide a compelling survey that will motivate the application of area production rates as meaningful measures of nonequilibrium behavior.

\item \textbf{Optimality:} Given an OU process, a weighted Frobenius norm of the empirical area production matrix, $\hat{\alpha}$, is the unique, optimal, test statistic for identifying violations of detailed balance among a class of observables that can be expressed as time-averaged path integrals against a linear vector field. Specifically, it maximizes the decay rate of the coefficient of variation of the path integral in the length of time observed. This proof is, to the author's knowledge, novel. We will show that this norm is equivalent to the entropy production rate for the process.

\end{enumerate}

%%%%%%%%%%%%%%%%%%%%%%%%%%%%%%%%%%%%%%%%%%
\section{Results}

%%%%%%%%%%% derivation and definition, relation to angular momentum of the probability fluid, currents and potential decomposition
\subsection{Essential Relations}

In this section, we will establish six key relations that characterize the commutator $\alpha$. In particular, given an OU process, the $\alpha$ is simultaneously:
\begin{enumerate}
\item the linear map that produces the rotational probability velocities and fluxes,
\item the angular momentum matrix of the probability fluid (up to a factor of 2), 
\item the linear map relating the conservative and rotational components of the vector field defining work (up to a factor of 2), 
\item the expected value of the stream function in two dimensions,
\item the time-averaged angular velocity of the process in a set of eigen-planes, and
\item the area production rate matrix.
\end{enumerate}

%%%%%%%%%%%
\noindent \textbf{How is $\alpha$ related to the circulation in probability currents?}
\vspace{0.05 in}

Suppose that $X(t)$ obeys an OU process (see equation \eqref{eqn: OU process SDE}). As before, let $\pi(\cdot,t)$ denote the density function for $X(t)$. The time evolution of $\pi(\cdot,t)$ is governed by the Fokker-Planck equation \eqref{eqn: OU Fokker-Planck}. If the process is initialized from a Gaussian distribution, then $\pi(\cdot,t)$ is Gaussian at all times, with covariance that evolves according to the matrix-valued ODE \cite{tomita1974irreversible}:
\begin{equation} \label{eqn: covariance dynamics}
    \frac{d}{dt} \Sigma(t) = D - (A \Sigma(t) + \Sigma(t) A^{\intercal}).
\end{equation}
It follows that the steady state covariance, $\Sigma_*$ must solve the Lyapunov equation \eqref{eqn: Lyapunov}.

Then, the probability fluxes at time $t$ are:
\begin{equation}
    j(x,t) = \left(\frac{1}{2} D - A \Sigma(t)  \right) \nabla \pi(x,t) = v(x,t) \pi(x,t)
\end{equation}
where $v(x,t)$ are the probability velocities:
\begin{equation}
    v(x,t) = - \left( \frac{1}{2} D - A \Sigma(t) \right) \Sigma(t)^{-1} x.
\end{equation}

These expressions recover the steady-state equations by substituting $\Sigma_*$ for $\Sigma(t)$. Just as we could use $v_*(x)$ to define the matrix $\alpha_*$, we can use $v(x,t)$ to define a matrix $\alpha(t)$. Here, we show that this matrix directs the circulating components of the probability fluxes and velocities, is, at each time, the angular momentum of the probability fluid, and converges to $\alpha_*$ at steady state.

Following \cite{tomita1974irreversible,tomita2008irreversible}, separate $A \Sigma(t)$ into its symmetric and skew-symmetric components:
\begin{equation}
\frac{1}{2} D - \frac{1}{2}(A \Sigma(t) + \Sigma(t) A^T) - \frac{1}{2}(A \Sigma(t) - \Sigma(t) A^T) =  \frac{1}{2}\frac{d}{dt} \Sigma(t) - \alpha(t) 
\end{equation}
where $\alpha(t) = \frac{1}{2}(A \Sigma(t) - \Sigma(t) A^T)$. We can then decompose the velocities and fluxes into dilational and rotational velocities and fluxes:
\begin{equation}
\begin{aligned}
& v(x,t) = v_d(x,t) + v_r(x,t) = -\frac{1}{2} \left[\frac{d}{dt}\Sigma(t) \right] \Sigma^{-1} x + \alpha(t) \Sigma^{-1}(t) x \\
& j(x,t) = j_d(x,t) + j_r(x,t) = \frac{1}{2} \left[\frac{d}{dt}\Sigma(t) \right] \nabla \pi(x,t) - \alpha(t) \nabla \pi(x,t) \\
& j_d(x,t) = v_d(x,t) \pi(x,t), \quad j_r = v_r(x,t) \pi(x,t).
\end{aligned}
\end{equation}
At steady state the covariance $\Sigma(t) = \Sigma_*$ stops changing so the dilational terms vanish and:
\begin{equation}
\begin{aligned}
& v_*(x) =  {v_r}_*(x) = \alpha_* \Sigma_*^{-1} x, \quad \alpha_* = \frac{1}{2}(A \Sigma_* - \Sigma_* A^{\intercal}) \\
& j_*(x) = {j_r}_*(x) = - \alpha_* \nabla \pi_*(x) = {v_r}_*(x) \pi_*(x).
\end{aligned}
\end{equation}

The ``rotational" fluxes and velocities are rotational in two senses: (1) they are divergence-free, and thus incompressible, (2) they are transverse, that is, they run perpendicular to the isoclines of $\pi(x,t)$. Therefore, they preserve the distribution. See, for example, the left panel of Figure \ref{fig: velocities and area production}, where the steady-state velocity vectors $v_*(x)$ are tangent, at all $x$, to the level sets of the steady-state density.

\begin{Lemma}
Given an OU process, initialized from a Gaussian distribution, the rotational velocities $v_r(x,t)$ and fluxes $j_r(x,t)$ are both divergence-free, $\nabla \cdot v_r(x,t) = \nabla \cdot j_r(x,t) = 0$ and transverse, $\nabla \pi(x,t) \cdot \nabla v_r(x,t) = \nabla \pi(x,t) \cdot \nabla j_r(x,t) = 0$ at all $x$ and $t$.
\end{Lemma}

\begin{proof}
First, the divergence of the rotational velocity $\nabla \cdot v_r(x,t) = \nabla \cdot \alpha(t) \Sigma^{-1}(t) x = \text{trace}(\alpha(t) \Sigma^{-1}(t))$. The matrix $\alpha(t)$ is skew-symmetric by construction, while $\Sigma^{-1}$ is symmetric since it is the inverse of a covariance. The trace of any skew-symmetric matrix times a symmetric matrix is always zero since the trace of a product of matrices is the matrix inner product, and the set of skew-symmetric matrices is orthogonal to the set of symmetric matrices under the matrix inner product. Therefore, $\nabla \cdot v_r(x,t) = 0$. 

The rotational velocities are transverse since $\nabla \pi(x,t) \cdot v_r(x,t) \propto x^{\intercal} \Sigma^{-1}(t) \alpha(t) \Sigma^{-1}(t) x = y^{\intercal} \alpha(t) y$ for some $y$. The matrix $\alpha(t)$ is skew-symmetric, so so all quadratic forms $y^{\intercal} \alpha(t) y = 0$. Therefore, $\nabla \pi(x,t) \cdot v_r(x,t) = 0$.

The fluxes are proportional to the velocities. It follows that they rotational flux $j_r$ is also perpendicular to $\nabla \pi(x,t)$ so is transverse. The divergence of the rotational fluxes is $\nabla \cdot v_r(x,t) \pi(x,t) = [\nabla \cdot v_r(x,t)] + v_r(x,t) \cdot \nabla \pi(x,t) = 0 + 0 = 0$. The first term is zero since the rotational velocities are divergence-free. The second is zero since they are transverse.
\end{proof}

Stepping back, the rotational fluxes/velocities describe the motion of the probability fluid along the isoclines of the probability distribution. These fluxes do not change the distribution since they are incompressible, and since the velocities corrspond to motion along the isoclines of the distribution. It follows that the remaining ``dilational" currents are responsible for changes in the shape of the distribution \cite{tomita1974irreversible}. 

%%%%%%%%%%%
\vspace{0.1 in}
\noindent \textbf{How is $\alpha(t)$ related to the angular momentum of the probability fluid?}
\vspace{0.05 in}

The angular momentum tensor, $L(t)$ for the probability fluid at time $t$ has $i,j$ entries:
\begin{equation} \label{eqn: angular momentum matrix}
    L(t)_{i,j} = \int_{x \in \mathbb{R}^d} [x_i,x_j] \times [j_i(x,t),j_j(x,t)] dx = \mathbb{E}_{X \sim \pi(\cdot,t)}\left[ x^{\intercal} R^{(i,j)} v(x,t)  \right]
\end{equation}
where $R^{(i,j)} = e_i e_j^{\intercal} - e_j e_i^{\intercal}$ is the unitary matrix that performs a ninety-degree rotation in the $i,j$ coordinate plane. Like area production rates, the angular momentum of the probability fluid is an indicator of nonequilibrium dynamics \cite{mellor2016characterization,zia2007probability}. Following \cite{weiss2019nonequilibrium}, we show that the $L(t)$ is proportional to $\alpha(t)$. 

\begin{Lemma}
    The angular momentum matrix $L(t) = 2 \alpha(t)$ and converges to $2 \alpha_*$ at steady-state.
\end{Lemma}

\begin{proof} First, expand the velocities in equation \eqref{eqn: angular momentum matrix}:
$$
\begin{aligned}
L(t)_{i,j} & = \mathbb{E}_{X \sim \pi(\cdot,t)}\left[ x^{\intercal} R^{(i,j)} v_d(x,t)  \right] + \mathbb{E}_{X \sim \pi(\cdot,t)}\left[ x^{\intercal} R^{(i,j)} v_r(x,t)  \right].
\end{aligned}
$$

We will show that the first integral is zero and the second returns $\alpha(t)$ via the following relation:
\begin{equation} \label{eqn: trace integral}
- \int_{x \in \mathbb{R}^d} x^{\intercal} M \nabla \mu(x) dx = \text{trace}(M)
\end{equation}
for any differentiable density $\mu$ that whose tails vanish exponentially or faster at infinity (see Appendix \ref{app: trace proof}). 

In particular:
$$
\begin{aligned}
        & \mathbb{E}_{X \sim \pi(\cdot,t)}\left[ x^{\intercal} R^{(i,j)} v_d(x,t)  \right] = \frac{1}{2}\int_{x \in \mathbb{R}^d} x^{\intercal} R^{(i,j)} \left[\frac{d}{dt} \Sigma(t) \right] \nabla \pi(x,t) dx \propto \text{trace}\left(R^{(i,j)} \frac{d}{dt} \Sigma(t) \right) \\ 
         & \mathbb{E}_{X \sim \pi(\cdot,t)}\left[ x^{\intercal} R^{(i,j)} v_r(x,t)  \right] = -\int_{x \in \mathbb{R}^d} x^{\intercal} R^{(i,j)} \alpha(t) \nabla \pi(x,t) dx = \text{trace}\left(R^{(i,j)} \alpha(t) \right) \\
\end{aligned}
$$

The matrix $R^{(i,j)}$ is skew-symmetric, and $\Sigma(t)$ is symmetric at all times, so the dilational term is the matrix inner product of a symmetric and skew-symmetric matrix. Therefore, the dilational term is zero. The rotational term is:
$$
\text{trace}\left(R^{(i,j)} \alpha(t) \right) = \sum_{k,l} R^{(i,j)}_{k,l} \alpha(t)_{k,l} = \alpha_{i,j}(t) - \alpha_{j,i}(t) = 2 \alpha_{i,j}(t).
$$

Therefore, $L(t) = 2 \alpha(t)$. Since $\alpha(t) \rightarrow \alpha_*$, $L(t)$ converges to $2 \alpha_*$ \end{proof}

%%%%%%%%%%%
\noindent \textbf{How is $\alpha_*$ related to work along trajectories and maximally probable trajectories?} 
\vspace{0.05 in}

Given an SDE of the form $dX(t) = \mu(X(t)) dt + G dW$ and with diffusion tensor, $D = \frac{1}{2} G G^{\intercal}$, define the vector field $w(x) = 2 D^{-1} \mu(x)$. This vector field plays a central role in the thermodynamic analysis of SDE's. Path integrals against $w(x)$ are proportional to the work needed to traverse the corresponding path in physically derived processes \cite{strang2020applications}. The vector field $w(x)$ also plays a central role in Friedlin-Wentzell quasi-potential analysis \cite{bender2013advanced,bressloff2014stochastic,ventsel1970small}. In particular, in a small-noise limit, the asymptotic probability of large fluctuations, and associated quantities, such as expected passage times, are characterized by trajectories that minimize work when integrated against $w(x)$ \cite{nolting2016balls,strang2019avoid,zhou2012quasi}.  

An SDE obeys detailed balance if and only if the vector field $w(x)$ can be expressed as the gradient of some potential function, $w(x) = -\nabla u(x)$ \cite{strang2020applications}. Such a potential only exists if $\partial_{x_i} \mu_j(x) = \partial_{x_j} \mu_i(x)$ for all $x$. If $\mu(x) = - A x$, as in an OU process, this requires that $[D^{-1} A]_{i,j} = [D^{-1} A]_{j,i}$, or $D^{-1} A = (D^{-1} A)^{\intercal}$. This symmetry condition is equivalent to the requirement that $\alpha_* = 0$ for OU processes, and can be shown by rearranging the algebra that led to equation \eqref{eqn: commutator for area production}.

If the SDE does not obey detailed balance, then $w(x)$ is not the gradient of any potential function, so is not conservative. Since any conservative vector field defines an equilibrium process, it is natural to seek decompositions that separate $w(x)$ into a conservative component associated with an equilibrium process, and a circulatory component. 

Two different decompositions are relevant here. If the process is diffusion-dominated, then the Helmholtz decomposition provides the most useful decomposition. For example, the circulating component of the Helmholtz decomposition can be used perturbatively to study systems near to equilibrium \cite{strang2020applications}. This expansion coincides with ``linear" thermodynamics, which studies systems that weakly violate detailed balance (c.f.~\cite{biot1958linear,garcia1991extended,onsager1931reciprocal1,onsager1931reciprocal2,schnakenberg1976network,yamamoto2016linear}). 

The Helmholtz decomposition expresses $w(x)$ as:
\begin{equation}
w(x) = w_c(x) + w_r(x) \text{ where } w_c(x) = -\nabla u(x) \text{ and } \nabla \cdot w_r(x) = 0. 
\end{equation}

To perform the Helmholtz decomposition, it suffices to find a potential $u(x)$ that such that the remainder $w_r(x) = w(x) - w_c(x)$ is incompressible. This leads to the Poisson equation:
\begin{equation} \label{eqn: Poisson}
\nabla \cdot D (\nabla u(x) + w(x)) = 0. 
\end{equation}

If the process is drift-dominated, as in a small-noise limit, then the relevant decomposition separates $w$ into a conservative component $w_c$ and a circulating component $w_r$ where $w_r$ is transverse to the conservative component \cite{nolting2016balls}. Then:
\begin{equation}
w(x) = w_c(x) + w_r(x) \text{ where } w_c(x) = - \nabla u(x) \text{ and } w_r(x) \perp w_c(x) \text{ for all } x.
\end{equation}

These constraints lead to a Hamilton-Jacobi equation \cite{nolting2016balls}:
\begin{equation} \label{eqn: Hamilton Jacobi}
\nabla u(x) \cdot D (\nabla u(x) + w(x)) = 0.
\end{equation}

Generically, these two decompositions differ. If there exists a potential $u$ that satisfies both the Poisson equation \eqref{eqn: Poisson} and the Hamilton-Jacobi equation \eqref{eqn: Hamilton Jacobi}, then that potential uniquely specifies the steady-state distribution. In particular, if $u(x)$ satisfies both the Poisson equation \eqref{eqn: Poisson} and the Hamilton-Jacobi equation \eqref{eqn: Hamilton Jacobi}, then:
\begin{equation}
\log(\pi_*(x)) = - \frac{1}{2} u(x)
\end{equation}
where the equality holds up to some additive constant.

This result can be checked by direct substitution into the Fokker-Planck equation. If $\pi_*(x) \propto \exp(- u(x))$ for some function $u(x)$, then setting $\nabla \cdot j_*(x) = 0$ is equivalent to setting $\nabla \cdot D (\nabla u(x) + w(x)) = \nabla u(x) \cdot D (\nabla u(x) + w(x))$. If $u(x)$ solves the Poisson and Hamilton-Jacobi equations, then both sides are zero. Indeed, if $-\log(\pi_*(x))$ satisfies either equation, then one side must be zero, and both sides must be equal, so $u_*(x) = -\log(\pi_*(x))$ must satisfy both equations. Any function of the form $u_*(x) = -\log(\pi_*(x))$ is the effective potential.

\begin{Lemma} Given an OU process, there exists a potential function $u(x)$ that satisfies both the Poisson and Hamilton-Jacobi equations, and it induces a decomposition $w(x) = w_c(x) + w_r(x)$ where $w_c(x) = - \nabla u(x)$, $w_r(x)$ is incompressible and divergence free, and:
\begin{equation}
    w_r(x) \propto \alpha_* w_c(x).
\end{equation}
\end{Lemma}

\begin{proof}
Given an OU process, $w(x) =- D^{-1} A x$, so is not $L_1$ or $L_2$ integrable. It follows that the Helmholtz decomposition is not unique. That said, there is a unique choice of the potential $u(x)$ that satisfies both the Poisson and Hamilton-Jacobi equations, thus equals the effective potential. 

Recall that the steady state of an OU process is Gaussian with covariance $\Sigma_*$ where $\Sigma_*$ solves the Lyapunov equation \eqref{eqn: Lyapunov}. Therefore, the effective potential is proportional to,
\begin{equation}
u_*(x) = \frac{1}{2} x^{\intercal} \Sigma_*^{-1} x.
\end{equation}

To show that this potential simultaneously solves the Poisson and Hamilton-Jacobi equations it suffices to show that it satisfies one of the two. The Poisson equation is:
$$
\nabla \cdot D (\nabla u_*(x) - 2 D^{-1} A x) = \nabla \cdot (D \Sigma_*^{-1} x - 2 A x) = \text{trace}(D \Sigma_*^{-1}) - 2 \text{trace}( A). 
$$

Therefore, the Poisson equation is satisfied if $\text{trace}(\Sigma_*^{-1}) = 2 \text{trace}(D^{-1} A)$. To show that $\Sigma_*$ satisfies this equality, multiply the Lyapunov equation \eqref{eqn: Lyapunov} by $\Sigma_*^{-1}$ on the right. Then:
$$
A + \Sigma_* A^{\intercal} \Sigma_*^{-1} = D \Sigma_*^{-1}.
$$
Evaluate a trace on both sides, and recall that the trace is preserved by similarity transforms and transposes. Then $\text{trace}(A) + \text{trace}(\Sigma_* A^{\intercal} \Sigma_*^{-1}) = 2 \text{trace}(A) = \text{trace}(D \Sigma_*^{-1}))$.

So, given an OU process, the effective potential function $u_*(x) = \frac{1}{2} x^{\intercal} \Sigma_* x$ provides a valid Helmholtz decomposition and Hamilton-Jacobi decomposition. This potential is meaningful. In particular, it coincides with the quasipotential, so, in a small noise limit, the expected hitting time between two states is exponential in their potential difference when moving uphill on the potential \cite{nolting2016balls,ventsel1970small,zhou2012quasi}. Since $u)*9x)$ solves both the Poisson and Hamilton-Jacobi equations, the remaining rotational component $w_r(x)$ responsible for nonequilibrium behavior is, like the rotational velocities and fluxes, both incompressible and transverse. 

The two components of $w(x)$ are:
\begin{equation}
w_c(x) = - \Sigma_*^{-1} x, \quad w_r(x) = w(x) - w_c(x) 
\end{equation}

Like the rotational velocities and fluxes, the rotational component can be expressed using $\alpha_*$. In particular:
$$
\begin{aligned} 
w_r(x) & = (-2 D^{-1} A + \Sigma_*^{-1} ) x = D^{-1} (-2 A \Sigma_* + D) \Sigma_*^{-1} x = (D - 2 A \Sigma_*) w_c(x) \\
& = (\Sigma_* A^{\intercal} - A \Sigma_* ) w_c(x) = -2 \alpha_* w_c(x). 
\end{aligned}
$$

Therefore:
\begin{equation} \label{eqn: rotational and circulatory via alpha}
w_r(x) = -2 \alpha_* w_c(x) = -L_* w_c(x) =  L_* \nabla u_*(x)
\end{equation}
where $2 \alpha_* = L_*$ is the steady-state angular momentum matrix responsible for the irreversible circulation of fluctuations. Thus, given an OU process, $\alpha_*$ is responsible for transforming between the conservative and rotational components of the vector-field that is integrated to define work and to identify most likely paths. \end{proof}

%% detailed balance remark
%Notice that, if the process obeys detailed balance, then $w(x) = - \nabla u(x)$ for some $u$, so the circulating component of both decompositions is automatically zero. It follows that both the Poisson and Hamilton-Jacobi equations can be solved simultaneously. Therefore, when in detailed balance, the steady state distribution always takes the Boltsmann form, $\pi_*(x) \propto \exp(-u(x))$ where $u(x)$ is any potential function satisfying $-\nabla u(x) = w(x) = D^{-1} \mu(x)$.

%% two d version (ratio of sizes)
Equation \eqref{eqn: rotational and circulatory via alpha} is particularly compelling in two dimensions. In two dimensions:
\begin{equation} \label{eqn: two-d area production}
\alpha_* = \omega R, \text{ where } R = \left[\begin{array}{cc} 0 &  1 \\ -1 & 0 \end{array} \right].
\end{equation}

Then, since $R$ is unitary, $\| w_r(x)\|_{2} = 2 \omega \| w_c(x) \|_2$ for all $x$. Rearranging, 
\begin{equation}
\alpha_* = \frac{1}{2} \frac{\|w_r(x)\|_2}{\|w_c(x)\|_2} R. 
\end{equation}

Thus, in two-dimensions, the entries of $\alpha_*$ express the ratio of the rotational component of $w$ to the conservative component of $w$ at all $x$. The larger $\alpha_*$, the larger $w_r$ relative to $w_c$. %If we express rhe Helmholtz potential using a scalar and vector potential, then  for a two-dimensional OU process, the two potentials are proportional, with scaling constant equal to $2 \omega$.

%%%%%%%%%%%
\vspace{0.1 in}
\noindent \textbf{How is $\alpha_*$ related to the stream function?}
\vspace{0.05 in}

In the previous section, we showed that $w(x) = D^{-1} A x$ can be expanded, $w(x) = w_c(x) + w_r(x) = w_c(x) - 2 \alpha_* w_c(x) = -(I - 2 \alpha_*) \nabla u_*(x)$ where $u_*(x) = \frac{1}{2} x^{\intercal} \Sigma_*^{-1} x$ is the effective potential which simultaneously solves the Poisson and Hamilton-Jacobi equations associated with a Helmholtz, and quasi-potential, decomposition. It follows that, the rotational component can be expanded:
\begin{equation} \label{eqn: rotational forces}
    w_r(x) = 2 \alpha_* \nabla u_*(x).  
\end{equation}

A vector field of the form \eqref{eqn: rotational forces}, if set as the flow of an ODE, defines a Hamiltonian dynamic. In particular, for a two-dimensional process:
\begin{equation}
    w_r(x) = 2 \omega R \nabla u_*(x) = 2 \omega \nabla \times u_*(x) = \nabla \times (2 \omega u_*(x))
\end{equation}
where $\nabla \times$ is the curl. In this setting, $2 \omega u_*$ is the stream function for the rotational component of the vector field $w$ \cite{skinner2021estimating}. The stream function is the unique function (up to a constant) such that, the curl of the stream function recovers $w_r$. Following \cite{skinner2021estimating}:

\begin{Lemma} \label{lem: stream function}
    In two-dimensions, the magnitude of $\alpha_*$, $\omega$, equals one-half the expected value of the stream function $u_*(X)$ when $X$ is drawn from the steady state distribution $\pi_*$.
\end{Lemma}

\begin{proof}
    It suffices to show that $\mathbb{E}_{X \sim \pi_*}[u_*(X)] = 1$. The expectation is,  $\mathbb{E}_{X \sim \pi_*}[u_*(X)] = \frac{1}{2} \mathbb{E}_{X \sim \pi_*}[X^{\intercal} \Sigma_*^{-1} X]$. Let $Y = \Sigma_*^{-1/2} X$. Then $\mathbb{E}_{X \sim \pi_*}[u_*(X)] = \frac{1}{2} \mathbb{E}_{Y \sim \mathcal{N}(0,I)}[\|Y\|^2]$. $\|Y\|^2$ is a $\chi^2$ random variable in $d$ dimensions, so has expectation equal to $d$. Then, when $d = 2$, $\mathbb{E}[u_*(X)] = 1$, so $\omega = \frac{1}{2} \mathbb{E}_{X \sim \pi_*}[2 \omega u_*(X)]$.
\end{proof}

%%%%%%%%%%%
\vspace{0.1 in}
\noindent \textbf{How is $\alpha_*$ related to angular velocity?}
\vspace{0.05 in}

The matrix $\alpha_*$ is also closely related to the expected rate at which sample trajectories orbit the origin given a two-dimensional OU process. This result generalizes to arbitrary dimension by noting that, the projection of an $d$ dimensional OU process onto any pair of the eigenvectors of $A$ produces a two-dimensional OU process. 

%% define winding angle and time-averaged angular velocity
Consider a two-dimensional stochastic OU process. Let $\Theta(t) - \Theta(0)$ denote the total angle traced out around the origin by $X(t)$. Then, define the time-averaged angular velocity:
\begin{equation}
    \hat{\omega}(\{X(t)\}_{t=0}^{T}) = \frac{1}{T} (\Theta(T)  - \Theta(0))
\end{equation}

We will show that $\lim_{T \rightarrow \infty} \mathbb{E}[\hat{\omega}(\{X(t)\}_{t=0}^{T})] = \omega$ where $\alpha_* = \omega R$. In this case, the estimator $\hat{\omega}(\{X(t)\}_{t=0}^{T})$ does not converge in probability to its expectation since the variance in the estimator diverges as $\|X(t)\|$ approaches zero. 

%% change coordinates
First, note that any positive definite change of coordinates preserves the long time expectation of $\hat{\omega}(\{X(t)\}_{t=0}^{T})$ since each complete rotation in the transformed coordinate system will map back to a complete rotation of the same sign in the original coordinate system. Consider the transform, $y = \Sigma_*^{-1/2} x$ that isotropizes the steady state distribution. Let $A^{(y)} = \Sigma_*^{-1/2} A \Sigma_*^{1/2}$ denote the drift matrix in the new coordinates. The diffusion matrix in the new coordinates is:
\begin{equation}
    D^{(y)} = \Sigma_*^{-1/2} D \Sigma_*^{-1/2} = \Sigma_*^{-1/2} (A \Sigma_* + \Sigma_* A^{\intercal}) \Sigma_*^{-1/2} = A^{(y)} + {A^{(y)}}^{\intercal} = 2 A_s^{(y)}
\end{equation}
where $A_s^{(y)} = \frac{1}{2}(A^{(y)} + {A^{(y)}}^{\intercal})$ is the symmetric part of $A^{(y)}$. 

In this coordinate system, $\Sigma_*^{(y)} = I$, so:
\begin{equation}
    \alpha_*^{(y)} = \frac{1}{2}(A^{(y)} I -  I {A^{(y)}}^{\intercal}) = A_a^{(y)}
\end{equation}
where $A_a^{(y)} = \frac{1}{2}(A^{(y)} - {A^{(y)}}^{\intercal})$ is the anti-symmetric part of $A^{(y)}$. We will soon show that $\alpha_*$ coincides with the area production rate for an OU process. Then, using the change of area formula for linear transformations:
\begin{equation}
\alpha_* = \text{det}(\Sigma_*)^{1/2} \alpha_*^{(y)}.
\end{equation}

As before, any skew-symmetric matrix in two-dimensions is proportional to $R$, the ninety-degree rotation matrix, so:
\begin{equation}
    \alpha_*^{(y)} = \omega^{(y)} R.
\end{equation}

Then, in the coordinates that isotropizes the steady state:
\begin{equation}
dY(t) = (A_s^{(y)} + \omega^{(y)} R) Y(t) dt + (2 A_s^{(y)})^{1/2} dW(t).
\end{equation}

To simplify, diagonalize $A_s^{(y)}$. Since $A_s^{(y)}$ is symmetric, it is unitarily diagonalizable, so $A_s^{(y)} = Q \Lambda Q^{\intercal}$ where $Q$ is an orthonormal matrix. Let $Z = Q^{\intercal} Y$. Then:
\begin{equation}
    dZ(t) = (\Lambda + \omega^{(y)} R) Z(t) dt + (2 \Lambda)^{1/2} dW(t)
\end{equation}
where $\alpha_*^{(y)} = \alpha_*^{(z)}$ since $Q^{\intercal} R Q = R$. Alternately, note that  $Q$ defines a unitary transformation, so the area produced in the $z$ coordinate system equals the area produced in the $y$ coordinate system, thus $\omega^{(y)} = \omega^{(z)}$. 

To find the expected angular velocity of $X(t)$ compute the expected angular velocity of $Z(t)$. These angular velocities coincide since completing a full rotation (moving 2 $\pi$ radians) in the coordinates $Z$ is the same as completing a full rotation in $X$. Therefore, even if a constant angular velocity in $Z$ does not translate to a constant angular velocity in $X$, when averaged over many complete rotations the total angle covered will be the same in both coordinates systems up to at most a difference of $\pi$ radians. This difference becomes negligible after many complete rotations, so the mean angular velocity $X(t)$, $Y(t)$ and $Z(t)$ coincide. 

%% use Ito's Lemma to compute the SDE in polar coordinates
To compute the expected angular velocity, convert from Cartesian coordinates $z_1, z_2$ to polar coordinates $r, \theta$. Let $r(z_1,z_2) = \sqrt{z_1^2 + z_2^2}$ and $\theta(z_1,z_2) = \tan^{-1}{\left(z_2/z_1\right)}$. Then, in polar coordinates $R(t) = r(Z(t))$, and $\Theta(t) = \theta(Z(t)).$
	
To find the equations governing the motion of $R(t)$ and $\Theta(t)$ we use It\^o's Lemma \cite{gardiner1985handbook}. It\^o's Lemma states that, if $Z(t)$ is a multi-dimensional stochastic process that obeys the SDE $dZ = \mu(Z) dt + G(Z) dW$ and $f(z)$ is a twice differentiable function then, if $F(t) = f(Z(t))$:
\begin{equation} \label{eqn: Ito lemma}
dF = \left(\nabla_z f(Z) \cdot \mu(Z)  + \frac{1}{2} \text{trace}\left(G(Z)^T H_f(Z) G(Z) \right) \right) dt + \nabla_z f(z) \cdot G(Z) dW
\end{equation}
where $H_f(z)$ is the Hessian of $f$ at $z$. 

Applying It\^o's lemma:
\begin{equation} \label{eqn: OU polar}
	\begin{aligned}
	dR(t) = & \lambda \left(\left(1 - \nu \cos{(2 \Theta(t))} \right) \frac{1}{R(t)} - \left(1 + \nu \cos{(2 \Theta(t))} \right) R(t) \right) dt \\ & + \sqrt{2 \lambda} \left( \sqrt{1 + \nu} \cos{(\Theta(t))} dW_1(t) +  \sqrt{1 - \nu} \sin{(\Theta(t))} dW_2(t)\right) \\
    d\Theta(t) = & \left(\omega^{(y)} - \left(1 - \frac{1}{R(t)^2} \right)\lambda \nu \sin(2 \Theta(t)) \right) dt \\ & + \frac{\sqrt{2 \lambda}}{R(t)}\left( -\sqrt{1 - \nu} \sin(\Theta(t)) dW_1(t) +  \sqrt{1 + \nu}\cos(\Theta(t)) dW_2(t) \right)
	\end{aligned}
\end{equation}
where $[\lambda_1,\lambda_2] = \lambda[1+ \nu, 1 - \nu]$

This equation simplifies dramatically when $\lambda_1 = \lambda_2$. Then $\nu = 0$, so:
\begin{equation} \label{eqn: symmetric OU polar}
    \begin{aligned}
    & dR(t) = - \lambda\left( R(t) - \frac{1}{R(t)} \right)  + \sqrt{2 \lambda} dW(t) \\
    & d\Theta(t) = \omega^{(y)} dt + \frac{\sqrt{2 \lambda}}{R(t)} dW'(t).
	\end{aligned}
\end{equation}
where $W$ and $W'$ are independent Weiner processes satisfying, $[dW(t),dW'(t)] = R(\Theta(t)) [dW_1(t), dW_2(t)]$ and where $R(\theta)$ is the $\theta$ degree rotation matrix. 

The symmetric process defined in equation \eqref{eqn: symmetric OU polar} is particularly easy to analyze. The radial dynamics decouple from the angular dynamics. The angle $\Theta(t)$ has expected increments $d\Theta(t) = \omega^{(y)}$ at all times, so the expected angular velocity is, at all times, equal to $\omega^{(y)}$. Notice that the noise variance in the angular dynamics depends on the distance to the origin, $R(t)$. The closer to the origin, the greater the noise variance. The noise variance diverges as $X(t)$ approaches the origin. As a result, the long-time average angular velocity, $\hat{\omega}(\{X(t)\}_{t=0}^T)$ will not converge in probability to its expectation. Nevertheless, $\mathbb{E}[d\Theta(t)] = \omega^{(y)} dt$, no matter the state of the process. Here we will show that this same property holds for a generic two-dimensional process provided $X$ is drawn from the steady-state distribution. 

%% show that the area production is the average angular velocity

\begin{Lemma} \label{lem: angular velocity} The expected instantaneous angular velocity, $\mathbb{E}[d\Theta(t)]/dt$ when $X \sim \pi_*$, equals $\text{det}(\Sigma_*)^{-1/2} \omega$ where $\alpha_* = \omega R$.
\end{Lemma}

\begin{proof}
First, convert into the $Z$ coordinate system. Then, since we adopted the coordinate system that isotropizes the steady state, $R$ is drawn from a $\chi^2$ distribution, and $\Theta$ is drawn uniformly on $[0,2 \pi]$ when $X \sim \pi_*$. 

Then:
$$
\begin{aligned}
    \mathbb{E}[d\Theta(t)]/dt & = \mathbb{E}_{R}\left[ \frac{1}{2 \pi}  \int_{\theta = 0}^{2 \pi} \frac{1}{dt}\mathbb{E}[d\Theta|\Theta = \theta,R = r] d\theta  \right] \\
    & =  \mathbb{E}_{R}\left[ \frac{1}{2 \pi} \int_{\theta = 0}^{2 \pi} \left(\omega^{(y)} - \left(1 - r^{-2} \right) \lambda \nu \sin(2 \theta) \right) d\theta \right] \\
    & = \mathbb{E}_{R}\left[\omega^{(y)} \right] = \omega^{(y)}
\end{aligned}
$$
where the second equality follows since $\sin(\theta)$ integrates to zero over $[0,2\pi]$.

Therefore, $Z(t)$ has instantaneous expected angular velocity $\omega^{(z)} = \omega^{(y)}$ at all $t$ provided $X(t) \sim \pi_*$. To convert back to $X$, recall that $\alpha_*^{(z)} = \alpha_*^{(y)} = \text{det}(\Sigma_*)^{-1/2} \alpha_*$. The theorem statement follows by setting $\alpha_* = \omega R$ for some $\omega$.
\end{proof}

%% argue this is an unstable estimator
At first glance, Lemma \ref{lem: angular velocity} suggests a simple estimator for $\alpha_*$. For any two-dimensional OU process, the expected instantaneous angular velocity of the process, when initialized from steady state recovers $\omega$, thus $\alpha_*$. Why not estimate this expected angular velocity by tracking the total angle swept out by long trajectories? 

In this case, even though the OU process is ergodic, the long-time average $\hat{\omega}(\{X(t)\}_{t=0}^{T})$ does not converge to its ensemble average in probability since the noise variance in the angle $\Theta$ diverges as the process approaches the origin. When $R(t)$ is sufficiently small, the process is likely to wander very rapidly around the origin, scrambling the phase information and producing large uncertainty in the angular distance traversed by the process. 

It is natural, then, to consider a weighted estimator, where the angle swept out at each time step is weighted by the distance from the origin. Weighting the angle swept out by the distance from the origin recovers the total area produced along the trajectory. Thus, while $\alpha_*$ recovers the expected instantaneous angular velocity when initialized from steady state, it is more stably estimated by tracking the rate at which the process produces area. In the next section we will prove that the empirical area production rate converges, along long trajectories, to $\alpha_*$ in probability.

%%%%%%%%%%%
\pagebreak
\noindent \textbf{How is $\alpha_*$ related to area production?}
\vspace{0.05 in}

In practice, if we do not know the pair of matrices $(A, G)$, then we cannot solve for $\alpha_*$ directly. Nevertheless, $\alpha_*$ can be estimated by evaluating the empirical area production rate $\hat{\alpha}(\{X(t)\}_{t=0}^T)$ over long trajectories. 

\begin{Lemma} \label{lem: area production}
    The empirical area production rate $\hat{\alpha}(\{X(t)\}_{t=0}^T)$ converges in probability to its expectation, $\alpha_*$ in the limit as $T \rightarrow \infty$. 
\end{Lemma}

\begin{proof}
The empirical estimator $\hat{\alpha}(\{X(t)\}_{t=0}^T)$ is a time-average of a path integral. Since $X(t)$ is an ergodic process, long time averages converge in probability to the expectation of the argument of the path-integral when sampling an initial condition from the steady state distribution. It follows that:
$$
\begin{aligned}
\lim_{T \rightarrow \infty} \hat{\alpha}_{ij}(\{X(t)\}_{t=0}^T)  \longrightarrow & \frac{1}{2} \mathbb{E}_{X \sim \pi_*}\left[ \mathbb{E}_{dX|X}\left[[X_i,X_j] \times [dX_i, dX_j] \right] \right] \\
& = \frac{1}{2}\mathbb{E}_{X \sim \pi_*}\left[ \mathbb{E}_{dX|X}\left[X^{\intercal} R^{(i,j)} dX \right] \right] = \frac{1}{2}\mathbb{E}_{X \sim \pi_*}\left[ X^{\intercal} R^{(i,j)} \mathbb{E}_{dX|X}\left[ dX \right] \right] \\& = \frac{1}{2}\mathbb{E}_{X \sim \pi_*}\left[ X^{\intercal} R^{(i,j)} A X \right] = \frac{1}{2}\mathbb{E}_{X \sim \pi_*}\left[ X^{\intercal} R^{(i,j)} A \Sigma_* \Sigma_*^{-1} X \right] \\& = \frac{1}{2} \text{trace}(R^{(i,j)} A \Sigma_*) = \frac{1}{2} \left( [A \Sigma_*]_{i,j} - [A \Sigma_*]_{j,i} \right) \\
& =  \left[\frac{1}{2}(A \Sigma_* - (A \Sigma_*)^{\intercal} \right]_{i,j} = \left[\frac{1}{2}(A \Sigma_* - \Sigma_* A^{\intercal} \right]_{i,j} \\
& = {\alpha_*}_{ij}
\end{aligned}
$$
where the expectation was reduced via equation \eqref{eqn: trace integral}.
\end{proof}

\begin{figure}[t]
\begin{centering}
\includegraphics[trim = 70 5 70 30, clip, scale = 0.45]{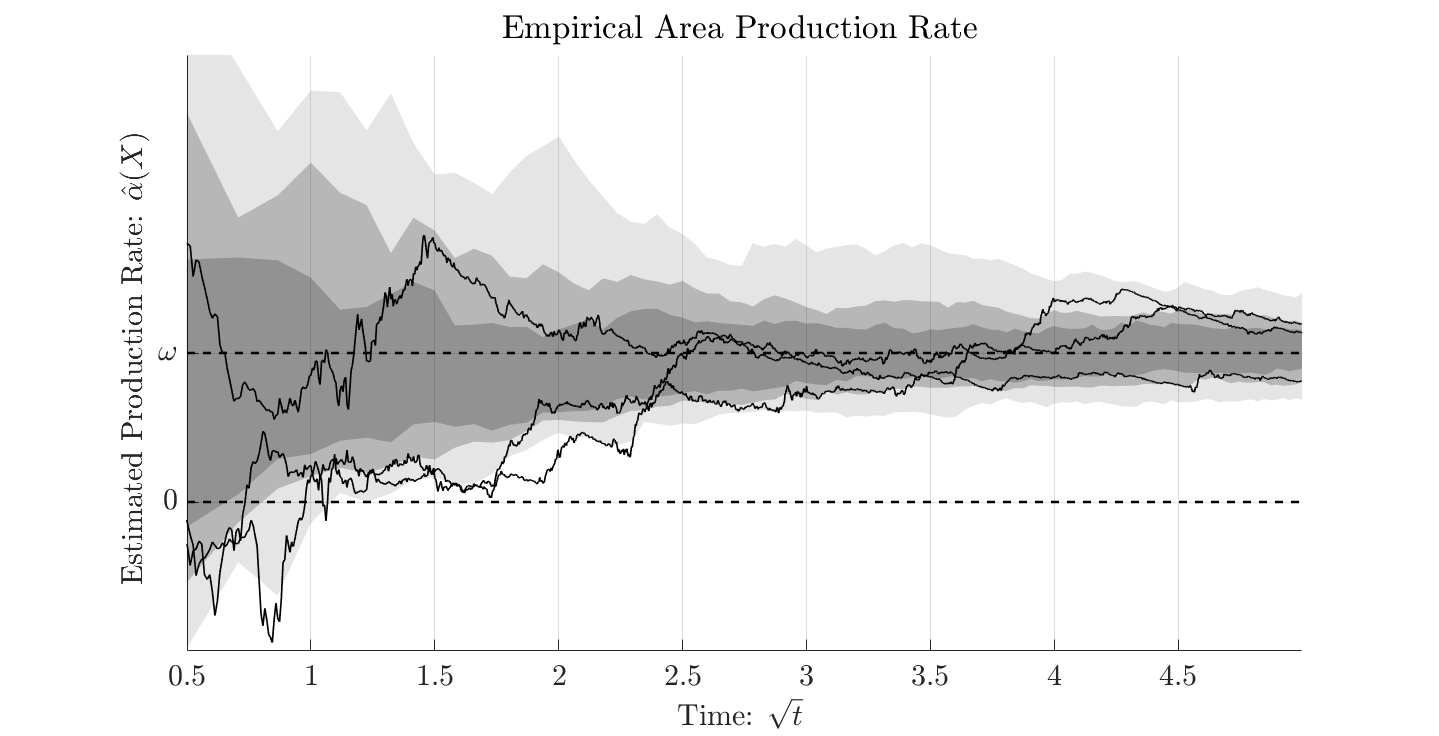}
\caption{Empirical area production rates, $\hat{\alpha}(X)$ converging to the true area production rate $\omega$ for a two-dimensional OU process. Trajectories for three independent samples are plotted in black. The gray shaded regions correspond to estimates for the 90\%, 70\%, and 50\% highest density intervals of the sampling distribution using 60 independent sample trajectories. %Notice that, up until time $t = 1$, $\hat{\alpha} = 0$ remains inside the 90\% interval, so, it is not until $t > 1$, that the estimated area production rates would provide significant evidence against detailed balance.
}
\end{centering}
\label{fig: area production estimate}
\end{figure}  

It follows that $\alpha_*$ is the matrix containing the rates at which the OU process traces out signed area in each coordinate plane. For this reason, we call $\alpha_*$ the \textit{area production rate} matrix. Lemma \ref{lem: area production} provides practical advice, since we can estimate $\hat{\alpha}(\{X(t)\}_{t=0}^T)$ with $\hat{\alpha}(\{X(t_k)\}_{k=0}^n)$ provided $\Delta t$ is sufficiently small. See Figure \ref{fig: area production estimate} for an example.

%In the next section we will show that, at least in two-dimensions, the area production rate matrix provides a satisfying description of the change in the geometry of the steady state distribution induced by rotation. 

%%%%%%%%%%% perturbative analysis in 2D
\subsection{Area Production and the Steady State in Two-Dimensions}

%% standardized coordinate system 
So far we have analyzed the area production rate matrix assuming a fixed OU process, and have expressed the area production rate matrix, $\alpha_*$, using the steady state covariance, $\Sigma_*$. Here, we show how, in two-dimensions, to solve for the area production rate without first solving for the steady state covariance. This solution leads to a reparameterization of the process in terms of the area production rate. We can, then, show how the steady state covariance $\Sigma_*$ changes in response to changes in the area production rate. 

To start, adopt a standardized coordinate system. Consider an OU process \eqref{eqn: OU process SDE} generated by the pair of matrices $(A, G)$. Let $y = D^{-1/2} x$ for any symmetric square root of $D$. Then:
\begin{equation}
    dY(t) = - A^{(y)} Y(t) dt + dW(t) \text{ where } A^{(y)} = D^{-1/2} A D^{1/2}.
\end{equation}
Then, $Y(t)$ obeys an OU process with isotropic noise. 

In the isotropized coordinates, $w^{(y)}(y) = A^{(y)} y$. Therefore, the isotropized process obeys detailed balance if and only if the drift matrix $A^{(y)}$ is symmetric. Changing coordinates via $G^{-1}$ is an invertible linear transformation, so the process $Y$ obeys detailed balance if and only if the original process, $X$, obeys detailed balance. Indeed, if $A^{(y)}$ is symmetric then $D^{-1} A = D^{-1} D^{1/2} A^{(y)} D^{-1/2} = D^{-1/2} A^{(y)} D^{-1/2}$ is symmetric. If $D^{-1} A$ is symmetric, then $A^{(y)} = D^{1/2} (D^{-1} A) D^{1/2}$ is symmetric. 

Consider, then, a symmetry, anti-symmetry decomposition of $A^{(y)}$:
\begin{equation}
    A^{(y)} = A_s^{(y)} - A_a^{(y)} \text{ where } \begin{cases} A_s^{(y)} = \frac{1}{2} (A^{(y)} - {A^{(y)}}^{\intercal}) \\ A_a^{(y)} = -\frac{1}{2} (A^{(y)} - {A^{(y)}}^{\intercal})
    \end{cases}
\end{equation}

Detailed balance requires $A_a^{(y)} = 0$. Therefore, any monotonically increasing function of the dimensionless ratio $\|A^{(y)}_a\|/\|A^{(y)}_s\|$ for some choice of matrix norm $\| \cdot \|$ is a reasonable signature of nonequilibrium behavior. We will show that, in two-dimensions, the area production rate takes exactly this form. 

The matrix $A^{(y)}_s$ is symmetric, so is unitarily diagonalizable. It follows that, there must exist a unitary matrix $Q$ such that $Q A^{(y)} Q^{-1} = \Lambda$ where $\Lambda$ is a diagonal matrix with all positive diagonal entries. Let $z = Q y$. Applying a unitary transformation to the Weiner increment produces a new Weiner increment so, in the rotated coordinate system:
\begin{equation} \label{eqn: OU in Z}
    dZ(t) = \left(- \Lambda + A_a^{(z)}\right) Z(t) dt + dW
\end{equation}
where $A^{(z)}_a = Q A_a^{(y)} Q^{-1}$ is the similarity transformation of the skew symmetric component. So, without loss of generality, we restrict our attention to OU processes of the form \eqref{eqn: OU in Z}:
\begin{equation} \label{eqn: OU form}
    dZ(t) = \left(- \Lambda + S \right) Z(t) dt + dW
\end{equation}
where $S$ is a skew-symmetric matrix.

In this coordinate system, $- \Lambda z$ defines a conservative vector field, while $S z$ defines an incompressible, thus rotational, vector field responsible for all violations of detailed balance. The larger $S$, the more strongly the process violates detailed balance. Therefore, we will consider the family of OU processes defined:
\begin{equation}
    dZ(t) = \left(- \Lambda + s \hat{S} \right) Z(t) dt + dW
\end{equation}
where $\hat{S} = S/\|S\|$ and where the parameter $s$ controls the degree of circulation introduced into the system. 

%% two-D solution for the area production
In two-dimensions, the parameter $s$ responsible for varying the degree of circulation in the system can be expressed in terms of the area production rate for the process. In two, dimensions the area production rate matrix is proportional to $R$, with area production rate $\omega$ (see equation \eqref{eqn: two-d area production}). In particular,
\begin{equation} \label{eqn: area production in 2d isotropized}
    \alpha_*^{(z)} = \alpha_*^{(y)} = \omega^{(y)} R \text{   where   } \omega^{(y)} = \frac{1}{2} \frac{a^{(y)}_{12} - a^{(y)}_{21}}{|a^{(y)}_{11} + a^{(y)}_{22}|} = \frac{1}{2} \frac{\|A_a^{(y)}\|_{1}}{\|A_s^{(y)}\|_{1}}
\end{equation}
and where $\|M\|_{1} = \sum_{i,j} |m_{i,j}|$ is the Schatten one norm. The area production rate in the original coordinate system can be recovered from the area production rate in the isotropized coordinate system via scaling by the determinant of the linear transformation responsible for changing between the coordinate systems:
\begin{equation} \label{eqn: area production in 2d}
    \alpha_* = \omega R \text{  where  } \omega = \text{det}(D)^{1/2} \omega^{(y)} = \frac{1}{2}\text{det}(D)^{1/2} \frac{a^{(y)}_{12} - a^{(y)}_{21}}{|\text{trace}(A)|}.
\end{equation}

In two dimensions, all skew symmetric matrices are proportional to $R$. If we adopt the infinity norm, then $\hat{S} = R$. Then, $A^{(z)}(s) = -\Lambda + s R$. Let $\lambda = \bar{\lambda}[1 + \mu, 1 - \mu]$ for $\bar{\lambda} > 0$ and for $\mu \in [0,1)$. %Then:
%
%$$
%A^{(z)} = \left[\begin{array}{cc} -\bar{\lambda} (1 + \mu) & s \\ -s & -\bar{\lambda} (1 - \mu) \end{array} \right] = \bar{\lambda} 
% \left[\begin{array}{cc} -(1 + \mu) & s/\bar{\lambda} \\ - s/\bar{\lambda} & -(1 - \mu) \end{array} \right] 
%$$
%
By equation \eqref{eqn: area production in 2d}, $\omega^{(z)} = \omega^{(y)} = \frac{1}{2} (s/\bar{\lambda})$ so:
\begin{equation} \label{eqn: standard parameterization}
    A^{(z)}(\bar{\lambda},\mu,\omega^{(y)}) = \bar{\lambda} \left[ \begin{array}{cc} -(1 + \mu) & 2 \omega^{(y)} \\ - 2 \omega^{(z)} & - (1 - \mu) \end{array} \right] = - \Lambda + 2 \bar{\lambda} \omega^{(y)} R
\end{equation}

Equation \eqref{eqn: standard parameterization} provides a standard parameterization for $A^{(z)}$ in terms of a scale parameter $\bar{\lambda}$ that controls the overall size of the drift term relative to the diffusive noise term, an eccentricity parameter $\mu$ that controls the difference in the strength of the restoring force pulling trajectories back towards the origin along the coordinate directions in $Z$, and the area production rate in the isotropized coordinate system, $\omega^{(z)} = \omega^{(y)} = \text{det}(D)^{-1/2} \omega$. Then, $\Sigma_*^{(z)}$ can be expressed as a function of the parameters $\bar{\lambda}, \mu, \omega^{(y)}$. Varying $\omega^{(y)}$ is equivalent to varying the degree of circulation, $s = 2 \bar{\lambda} \omega^{(y)}$ driving the system. So, from now on we will parameterize our OU process by first fixing $\bar{\lambda}, \mu, \omega^{(y)}$, then by fixing the matrices $Q$ and $D^{1/2}$ responsible for converting from $z$ to $y$ to $x$. These transformations are automatic, in particular, $\Sigma_*^{(x)} = D^{1/2} \Sigma_*^{(y)} D^{1/2} = D^{1/2} Q \Sigma_*^{(z)} Q D^{1/2}$, so it suffices to study the behavior of the OU process in the standardized coordinate system. 

\vspace{0.1 in}
\noindent \textbf{How is $\alpha$ related to perturbations of the steady state distribution away from equilibrium?}
\vspace{0.05 in}

\begin{figure}[t]
\includegraphics[trim = 135 50 100 30, clip,width = \textwidth]{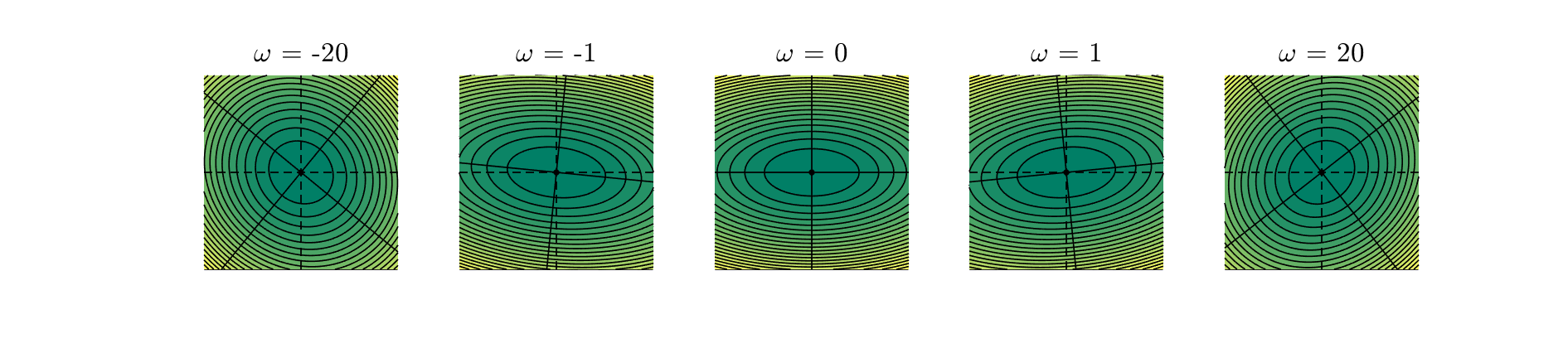}
\caption{Contours of the steady state distribution for varying area production rates $\omega$. The solid black lines are the principal axes (singular vectors) of the steady state covariance, $\Sigma_*(\omega)$. The dashed black lines are the principal axes at equilibrium. Notice that, increasing $\omega$ tilts the steady state away from its equilibrium orientation. The direction and magnitude of the tilt is determined by the sign and magnitude of $\omega$. The tilt increases quickly for $\omega$ near zero and slowly once $\omega$ is large. The contours of the steady-state distribution become less eccentric as the magnitude of $\omega$ increases.}
\label{fig: steady states}
\end{figure}

%% insert demonstration figure

%% perturbation problem
 Consider the canonical process $dZ(t) = (- \Lambda + s S) Z(t) dt + dW(t)$. How does the steady state distribution $\pi_*^{(z)}$ vary as a function of $s$ for fixed $\Lambda$ and $S$? Since the steady state is necessarily Gaussian with mean zero, it suffices to solve for the steady state covariance as a function of $s$. This amounts to solving for the matrix valued function $\Sigma_*(s)$ that solves the Lyapunov equation:
 \begin{equation} \label{eqn: standardized Lyapunov equation}
     (-\Lambda + s S) \Sigma_*(s) + \Sigma_*(s) (- \Lambda + s S)^{\intercal} = I.
 \end{equation}

 We will solve this equation assuming that the process is two-dimensional, using the $\bar{\lambda}, \mu, \omega$ parameterization. We focus on the two-dimensional case since, in two-dimensions, the steady state covariance $\Sigma_*(s)$ can be expressed as a simple geometric transformation of the equilibrium covariance $\Sigma_*(0)$, where the transformation is expressed in terms of the area production rate. Moreover, in a variety of experimental settings, only two variables are available to the observer (c.f.~\cite{skoge2013chemical,suter2011mammalian}). Similar results partially generalize to higher dimensions. We save that exposition for other work.

 Figure \ref{fig: steady states} shows a sequence of steady-state distributions corresponding to different values of $\omega$, and thus, of $s$, in a two-dimensional OU process with fixed $\Lambda$. Notice that the steady state distribution tilts away from the steady state at equilibrium by an angle that is monotonically increasing in $\omega$. Also notice that the contours of the steady state distribution become less eccentric as the magnitude of $\omega$ increases.  

%% geometric solution (proof to appendix)
\begin{Theorem} \label{thm: steady state ellipse geometry}
    Consider a two-dimensional OU process of the form $dZ = (-\Lambda + s R) Z(t) + dW(t)$ where $\lambda = \bar{\lambda}[(1 + \mu), (1 - \mu)]$, and where $s = 2 \bar{\lambda} \omega$. Let $\Sigma_*(\omega)$ denote the corresponding steady state covariance as a function of area production rate, and $\Sigma_*(0) = \frac{1}{2} \Lambda^{-1}$ denote the steady state at equilibrium. Then the steady state covariance may be expressed as the following geometric transformation of $\Sigma_*(0)$:
    \begin{enumerate}
        \item The singular vectors of $\Sigma_*(\omega)$ are rotated by the tilt angle $\theta(\omega) = \frac{1}{2} \tan^{-1}(2 \omega)$ with respect to the singular vectors of $\Sigma_*(0)$ (the coordinate axes).
        \item Define the ellipse function $\mathcal{E}(M) = \{z | z^{\intercal} M z = 1\}$. %Then the level sets of the steady state density $\pi_*^{(z)}$ are concentric ellipses proportional to $\mathcal{E}(\Sigma_*^{-1})$. 
        \item Define an inner ellipse $\mathcal{E}(\Lambda^+)$ where $\Lambda^+ = \tfrac{\bar{\lambda}}{2} \text{diag}((1 + \mu), 1)$. Then $\mathcal{E}(\Lambda^+)$ is the ellipse with minor axis equal to the minor axis of $\Sigma_*(0)$, and with major axis of length equal to the harmonic average of the major and minor axes lengths of $\Sigma_*(0)$. 
        \item Define an outer ellipse $\mathcal{E}(\Lambda^-)$ where $\Lambda^+ = \tfrac{\bar{\lambda}}{2} \text{diag}(1, (1 - \mu))$. Then $\mathcal{E}(\Lambda^-)$ is the ellipse with major axis equal to the major axis of $\Sigma_*(0)$, and with minor axis of length equal to the harmonic average of the major and minor axes lengths of $\Sigma_*(0)$. 
        \item Then, $\Sigma_*(\omega)$ is the unique covariance matrix whose ellipse $\mathcal{E}(\Sigma_*^{-1}(\omega))$ has principal axes tilted by $\theta(\omega)$ that is entirely contained inside of $\mathcal{E}(\Lambda^-)$, is tangent to $\mathcal{E}(\Lambda^-)$ at exactly two locations, entirely contains $\mathcal{E}(\Lambda^+)$, and is tangent to $\mathcal{E}(\Lambda^+)$ at exactly two locations.
    \end{enumerate}

    The ellipse $\mathcal{E}(\Sigma_*^{-1}(\omega))$ intersects the inner and outer ellipses along rays formed by rotating the coordinate axes by twice the tilt angle, $2 \theta(\omega) = \tan^{-1}(2 \omega)$. These rays form an orthonormal basis parallel to the vector $[1, 2 \omega]$. (See Figure \ref{fig: ellipses}.) 
\end{Theorem}

%% insert schematic and demonstration figure

%% corollaries
\begin{Corollary}[Tightest Bounding Distributions] \label{corr: bounding distribution}
Under the assumptions of Theorem \ref{thm: steady state ellipse geometry}, the pair of density functions corresponding to the mean zero Gaussian distributions with inverse covariances $\Lambda^+$ and $\Lambda^-$ formed by either replacing $\text{min}(\lambda)$ or $\text{max}(\lambda)$ with $\text{mean}(\lambda) = \bar{\lambda}$, provide the tightest upper and lower bounds on the steady-state density $\pi_*(x;\omega)$ that is valid for all $x$ and all $\omega$.
\end{Corollary}

\begin{Corollary}[Singular Value Decomposition] \label{corr: SVD}
      Under the assumptions of Theorem \ref{thm: steady state ellipse geometry}, the singular values of the steady state covariance are:
        \begin{equation} \label{eqn: singular values}
            \frac{1}{2 \bar{\lambda}}\left(1 \pm \sqrt{\frac{\mu^2}{1 + (2 \omega)^2}}\right)^{-1}
        \end{equation}
    and the singular vectors are formed by rotating the coordinate axes by $\theta(\omega) = \frac{1}{2} \tan^{-1}(2 \omega)$. 

\end{Corollary}

\begin{Corollary}[Monotonicity] \label{corr: monotonicity} Under the assumptions of Theorem \ref{thm: steady state ellipse geometry}: (1) the trace of the inverse covariance, $\text{trace}(\Sigma_*^{-1}(\omega))$, is constant in $\omega$ and equals $2 \bar{\lambda}$, (2) the determinant of the covariance $\Sigma_*(\omega)$ is monotonically decreasing in $|\omega|$, and (3) the eccentricity in the ellipse $\mathcal{E}(\Sigma_*^{-1}(\omega))$ is monotonically decreasing in $|\omega|$.
\end{Corollary}

\begin{Corollary}[Steady State Entropy] \label{corr: entropy}
    Under the assumptions of Theorem \ref{thm: steady state ellipse geometry}, the steady state entropies of $Z$, and any linear transformation of $Z$ (e.g.~$x = D^{1/2} Q z$), are monotonically decreasing in the area production rate $\omega$, and are minimized in the limit as $\omega$ diverges. 
\end{Corollary}

\begin{Corollary}[Explicit Form] \label{corr: form}
    Under the assumptions of Theorem \ref{thm: steady state ellipse geometry}, the steady state covariance has the explicit form:
    \begin{equation} \label{eqn: explicit solution for covariance}
        \Sigma_*(\omega) = \frac{1}{2 \bar{\lambda}} \left(1 - \frac{\mu^2}{1 + (2 \omega)^2} \right)^{-1} \left( I - \frac{\mu}{1 + (2 \omega)^2} \left[ \begin{array}{cc} 1 & 2 \omega \\ 2 \omega & -1 \end{array}\right]\right).
    \end{equation}
\end{Corollary}

\begin{Corollary}[Limits] \label{corr: limits}
Under the assumptions of Theorem 1, $\lim_{\omega \rightarrow 0} \Sigma_*(\omega) = \Sigma_*(0) = \frac{1}{2} \Lambda^{-1}$, and $\lim_{\omega \rightarrow \infty} \Sigma_*(\omega) = (2 \bar{\lambda})^{-1} I$, so the steady state distribution has circular level sets and samples from the steady state have independent components in the strong rotation limit. If $x = D^{1/2} Q z$ for some unitary matrix $Q$, then $\lim_{\omega \rightarrow \infty} \Sigma_*^{(x)}(\omega) = (2 \bar{\lambda})^{-1} D$. 
\end{Corollary}

%% proof
\begin{proof}[Proof Outline for Theorem 1 and Corollaries 1 - 6]
    The complete proofs for Theorem \ref{thm: steady state ellipse geometry}, and Corollaries \ref{corr: bounding distribution} - \ref{corr: limits} are provided in Appendix \ref{app: steady state proofs}. Direct substitution of \eqref{eqn: explicit solution for covariance} into the Lyapunov equation \eqref{eqn: Lyapunov} for $A$ parameterized according to \eqref{eqn: standard parameterization} establishes Corollary \ref{corr: form}. Corollary \ref{corr: monotonicity} follows directly by analyzing the behavior of the singular values as functions of $\omega$. Corollary \ref{corr: entropy} follows since the entropy of any Gaussian distribution is, up to the addition of a dimensional constant, equal to the log determinant of its covariance. By corollary \ref{corr: monotonicity}, the determinant of the steady state covariance is monotonically decreasing in $\omega$, so the entropy decreases as well. The entropy of any linear transformation of a random variable equals the entropy of that variable plus a constant that depends on the transform. Corollary \ref{corr: limits} follows immediately by evaluating the limits as $\omega$ approaches 0 or infinity. Corollary \ref{corr: SVD} follows by verifying that the matrix with singular values \eqref{eqn: singular values} and singular vectors rotated by the angle $\theta(\omega)$ recovers the explicit covariance presented in Corollary \ref{corr: form}. 

    Theorem \ref{thm: steady state ellipse geometry} follows by showing that, the ellipse specified by a matrix with singular value decomposition as specified in Corollary \ref{corr: SVD} has principal axes  tilted by $\theta(\omega)$ degrees, intersects the outer and inner ellipses at pairs of orthogonal points, and is tangent to those ellipses where it intersects them. This proof is accomplished by solving a pair of constrained optimization problems. Namely, optimize $u_*(z; \omega) = \frac{1}{2} z^{\intercal} \Sigma_*^{-1}(\omega) z$ constrained to $z$ in either the inner or outer bounding ellipse. These optimization problems can be converted into a pair of eigenvalue problems. The solutions verify the theorem claim, are orthogonal to one another, and are rotated by twice the tilt angle, $2 \theta(\omega)$, off the coordinate axes. Corollary \ref{corr: bounding distribution} is a necessary consequence of the relationship between the elliptical level sets of a Gaussian distribution and its density function. 
\end{proof}

%% visualization and discussion
\begin{figure}[t]
\centering
\includegraphics[trim = 100 30 100 30, clip,scale = 0.47]{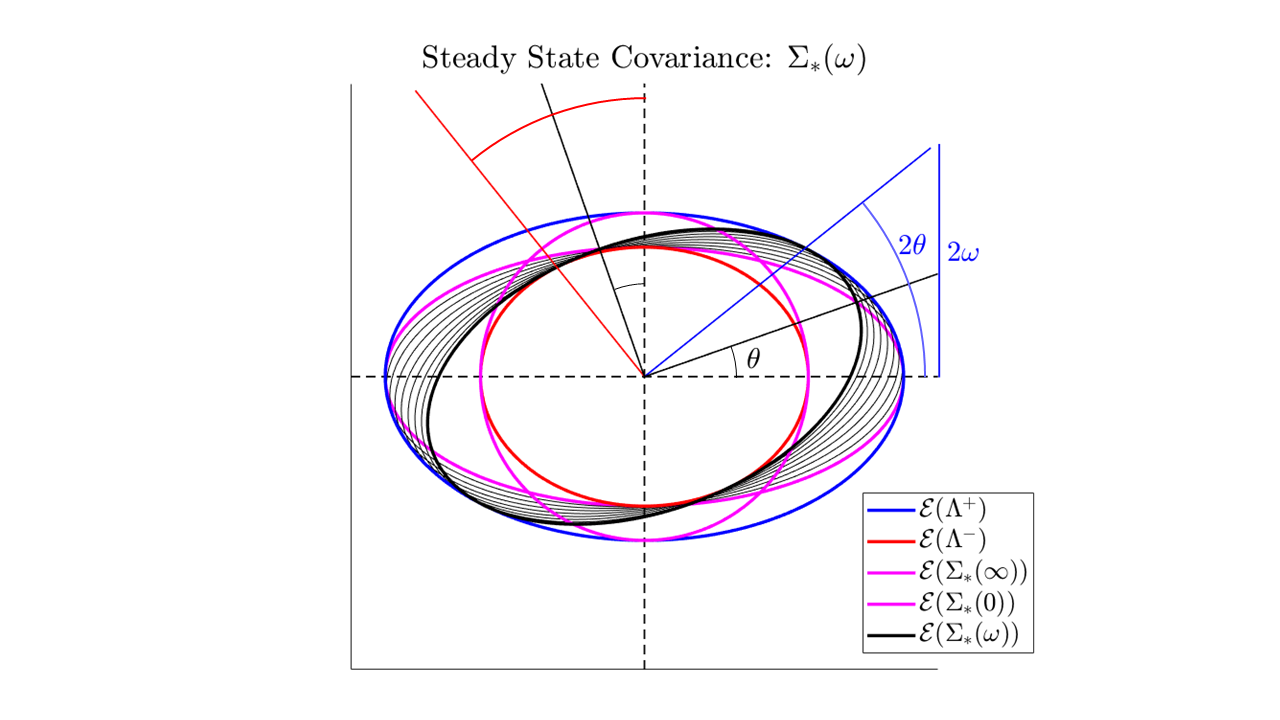}
\caption{Ellipses corresponding to the steady state density as a function of $\omega$. The dashed black lines are the coordinate axes, and the principal axes of $\Sigma_*$ at equilibrium. The bold magenta oval that is stretched out along the horizontal axis corresponds to the equilibrium density and covariance, $\Sigma_*(0)$. It has axes of length $\frac{1}{2} \lambda_-^{-1}$ and $\frac{1}{2} \lambda_+^{-1}$ where $\lambda_{\pm} = \bar{\lambda}(1 \pm \mu)$. The thin black ellipses represent $\Sigma_*(\omega)$ for increasing values of $\omega$. The bold black ellipse corresponds to $\omega = 4$. The steady state ellipse corresponding to $\Sigma_*(\omega)$ is rotated off the coordinate axes by $\theta(\omega) = \tfrac{1}{2} \tan^{-1}(2 \omega) \simeq \omega$. The size of the ellipse (lengths of its principal axes) is determined by a pair of bounding ellipses shown in blue (outer) and red (inner). The ellipses are constructed by first drawing a circle, corresponding to the infinite rotation limit, centered at the origin with radius determined by the average eigenvalue, $\frac{1}{2} \bar{\lambda}^{-1}$. Next draw the outer ellipse (bold, blue) by finding the smallest ellipse that contains both the circle and the equilibrium ellipse. Draw the inner ellipse (bold, red) by finding the largest ellipse contained inside both the circle and the equilibrium ellipse. For any particular $\omega$, the steady-state ellipse intersects the outer and inner ellipses at four orthogonal points and is tangent to the bounding ellipses there. These intersections occur along rays tilted $2 \theta(\omega)$ off the coordinate axes, so may be recovered by following all ninety degree rotations of the vector $[1,2 \omega]$.   }
\label{fig: ellipses}
\end{figure}

The geometric rules stated in Theorem \ref{thm: steady state ellipse geometry} are sensible. At equilibrium, the circulating component vanishes, so the process obeys detailed balance. When in detailed balance, the steady state is always a Boltsmann distribution with potential equal to any function such that $-\nabla u(x) = w(x) = - \Lambda x$. The potential $u(x) = \frac{1}{2} x^{\intercal} \Lambda x$ suffices, so $\Sigma_*(0) = \frac{1}{2} \Lambda^{-1}$. 

As $\omega$ increases, the drift field $w(x) = (-\Lambda + 2 \bar{\lambda} \omega R)$ induces rotation. This will skew the steady-state distribution in the direction of the applied rotation. The resulting skew rotates the principal axes by an angle controlled by the strength of the applied roation, and roughly proportional to the applied rotation when it is small. In particular, $\theta(\omega) = \frac{1}{2} \tan^{-1}(2 \omega) \simeq \omega$ for small $\omega$.

Increasing $\omega$ drives trajectories to orbit the origin with winding angle produced at an average rate equal to $\omega$.\footnote{Like the stochastic area, probability currents, and angular momentum, the winding velocity has also been proposed as a test statistic for detecting violations of detailed balance \cite{gradziuk2019scaling}. To show that the winding velocity equals the area production rate in an OU process, change into polar coordinates, then applies It\^o's lemma to derive the stochastic process governing the phase.} For sufficiently large $\omega$, the circulating drift is much larger than the conservative attraction to the origin, so sample trajectories complete many orbits before drifting on the equilibrium potential $-\frac{1}{2} z^{\intercal} \Lambda z$. In this case, the time scales for circulation and inward attraction separate. Then, over the time scale at which the process drifts in or out, it completes many orbits, so averages the different attractive forces felt along separate directions. As a result, the anisotropy in the steady state distribution diminishes, and the eccentricity in the covariance decreases, as $\omega$ increases. The eccentricity vanishes as $\omega$ diverges. When $\omega$ is very large, circulation dominates, and the different attractive forces average, so the steady state is spherically symmetric with variance proportional to one over the average of $\lambda_+$ and $\lambda_-$, $\bar{\lambda}$. 

The strongly geometric character of this solution, and its simple dependence on $\omega$ recommends $\omega$ as a measure of the intensity of nonequilibrium circulation in two-dimension OU processes.

%%%%%%%%%%% linear observables, production rate and variance in observed production rates, controlled by area production matrix, maximal CV growth rate, optimized by measuring the Frobenius norm of the area production matrix
\subsection{Optimal Observables}

So far we have argued for the adoption of the area production rate matrix as a unifying signature of nonequilibrium circulations in OU processes by showing that it is meaningful. Here we adopt a different approach. Instead of demonstrating that $\alpha_*$ admits a rich interpretation, we argue that it is, in a sense, directs the optimal choice, when selecting a test statistic for identifying violations of detailed balance. To establish this result we will show that $\alpha$ controls the production rate of an entire class of observables, then, that the coefficient of variation in an appropriate norm of $\alpha$ decreases the fastest among all observables in that class. 

\pagebreak
\noindent \textbf{How is $\alpha$ related to the production rate of other observables?}
\vspace{0.05 in}

%% generic linear observable
First, consider an observable of the form:
\begin{equation} \label{eqn: observable}
    \hat{B}(\{X(t)\}_{t=0}^T) = \int_{t=0}^T b(X(t)) \cdot dX(t)
\end{equation}
where $b: \mathbb{R}^d \rightarrow \mathbb{R}^d$ is a vector field on $\mathbb{R}^d$. An observable of this form is an example of a stochastic line integral \cite{teitsworth2022stochastic} or generalized current \cite{gingrich2016dissipation}.  Notice that the empirical area produced by a trajectory on the $i,j$ plane is a member of the class of observables, with $b(x) = R^{(i,j)} x$. 

Then, define the empirical, time-averaged production rate:
\begin{equation}
\hat{\beta}(\{X(t)\}_{t=0}^T)) = \frac{1}{T} \hat{B}(\{X(t)\}_{t=0}^T).
\end{equation}

Since $X$ is ergodic, the empirical, time-averaged production rate of the observable $\hat{\beta}(\{X(t)\}_{t=0}^T))$ converges in probability to its steady-state expectation provided $b$ is chosen so that the variance in the production rate is finite at steady-state \cite{boltzmann1910vorlesungen,feller1991introduction}. Assuming $b$ is chosen so that the variance is finite,
\begin{equation}
\hat{\beta}(\{X(t)\}_{t=0}^T) \longrightarrow \mathbb{E}_{X \sim \pi_*}\left[ \mathbb{E}_{dX|X}[b(X) \cdot dX] \right] = \beta
\end{equation}

Further, we will restrict our attention to stochastic line integrals whose vector field $b(x) = M x + v$ for some matrix $M \in \mathbb{R}^{d \times d}$ and vector $v \in \mathbb{R}^d$. Any observable of this form is a linear observable in the sense that is defined by path integration over a linear vector field. Examples include the entropy production rate of a process (the rate at which it exchanges energy with its reservoir), which uses $b(x) \propto w(x) = D^{-1} A x$ \cite{qian2002thermodynamics,tomita2008irreversible}. 

%% expectation in terms of alpha
\begin{Lemma} \label{lem: production rates}
    The expected production rate, $\beta$, of any observable defined as a path integral against a linear vector field $b(x) = M x + v$ along the path $\{X(t)\}_{t=0}^T$ is:
    \begin{equation}
    \beta = \text{trace}(M_a^{\intercal} {\alpha_*}) = \langle M_a, {\alpha_*} \rangle = \sum_{i,j} {m_{a}}_{ij} {\alpha_*}_{ij} =  \sum_{i < j} (m_{ij} - m_{ji}){\alpha_*}_{ij}
    \end{equation}
    where $M_a$ is the skew-symmetric component of $M$, $\alpha$ is the area production rate matrix, and $\langle \cdot, \cdot \rangle$ is the matrix inner product.
\end{Lemma}

\begin{proof}
Consider an observable $\hat{B}(\{X(t)\}_{t=0}^T)$ as defined in equation \eqref{eqn: observable} where the vector field $b(x) = M x + v$ for some $M, v$. Let:
\begin{equation}
M = M_s + M_a \text{  where  } \begin{cases} M_s = \frac{1}{2}(M + M^{\intercal}) \\
M_a = \frac{1}{2}(M - M^{\intercal})
\end{cases}
\end{equation}
	
Since $M_s$ is symmetric $M_s x + v$ can be rewritten as the gradient of the function $\phi(x) = \frac{1}{2} x^T M_s x + v^T x$. The path integral against the gradient of a function is just the difference in the value of the function at its end points. Therefore, if we break the observable into its symmetric and antisymmetric parts $\hat{B}(\{X(t)\}_{t=0}^T) =  \hat{B}_s(\{X(t)\}_{t=0}^T) + \hat{B}_a(\{X(t)\}_{t=0}^T)$ then $\hat{B}_s(\{X(t)\}_{t=0}^T) = \phi(X(T)) - \phi(X(0))$. Then, $\mathbb{E}[\hat{B}_s(\{X(t)\}_{t=0}^T)] = \mathbb{E}[\phi(X(T))] - \mathbb{E}[\phi(X(0))]$. If $X(0)$ is drawn from the steady state, then $X(T)$ is drawn from the steady state, so the two expectations cancel. If not, then, after a long time this expected value converges to $\mathbb{E}_{X \sim \pi_*}[\phi(X)] - \mathbb{E}_{X \sim \pi(\cdot,0)}[\phi(X)]$. This does not change in time, so does not contribute to the long-term production of the observable \cite{teitsworth2022stochastic}. It follows that the long time production rate of the observable only depends on the antisymmetric part.

Then, replacing a temporal average with a spatial average \cite{teitsworth2022stochastic}:
$$
\begin{aligned}
\beta & = \mathbb{E}_{X \sim \pi_*}[X^{\intercal} M_a^{\intercal} \mathbb{E}_{dX|X}\left[dX \right]] = \mathbb{E}_{X \sim \pi_*}[X^{\intercal} M_a A X] \\
& = \mathbb{E}_{X \sim \pi_*}[X^{\intercal} M_a A \Sigma_* \Sigma_*^{-1} X] = -\int_{x \in \mathbb{R}^d} x^{\intercal} M_a A \Sigma_* \nabla \pi_*(x) dx \\
& = -\text{trace}(M_a A \Sigma_*) = \text{trace}(M_a^{\intercal} A \Sigma_*) \\& = \langle M_a,A \Sigma_* \rangle.
\end{aligned}
$$
	
The matrix inner product can be simplified by rewriting $A \Sigma = \frac{1}{2}(A \Sigma_* + \Sigma_* A^T) + \frac{1}{2} (A \Sigma_* - \Sigma_* A^T) = \frac{1}{2} D + {\alpha_*}$. Symmetric and skew-symmetric matrices are orthogonal under the matrix inner product, so $\langle M_a, D \rangle = 0$. Therefore $\beta = \langle M_a, {\alpha_*} \rangle$.

The matrix inner product $\langle M_a, \alpha \rangle$ can be broken into a sum over each pair of coordinate planes $i,j$ by writing out the trace of the product explicitly:
$$
\begin{aligned}
\beta = \langle M_a, {\alpha_*} \rangle = \sum_{ij} m_{a_{ij}} {\alpha_*}_{ij} & = \sum_{i < j} m_{a_{ij}} {\alpha_*}_{ij} + m_{a_{ji}} {\alpha_*}_{ji} \\
& = \sum_{i < j} m_{a_{ij}} {\alpha_*}_{ij} + m_{a_{ij}} {\alpha_*}_{ij}  =  \sum_{i < j} 2 {m_a}_{ij} {\alpha_*}_{ij} = \sum_{i < j} (m_{ij} - m_{ji})\alpha_{ij}.
\end{aligned}
$$\end{proof}

Thus, the expected production rate of all linear observables is controlled by a matrix inner product against the area product rate matrix, $\alpha_*$. It follows that, if $\alpha_*$ is known, or estimated, then we can predict the expected production rate of any linear observable. 

%% variational characterization of alpha
\begin{Corollary}
    Given an OU process, no linear observable defined as a path integral against a linear vector field $b(x) = M x + v$ is produced faster than $\|\alpha_*\|_{\text{Fro}} \|M_{a}\|_{\text{Fro}}$. This upper bound is tight and is satisfied by any observable with $M_a \propto \alpha_*$.
\end{Corollary}

\begin{proof}
    By Lemma \ref{lem: production rates}, the expected production rate is the matrix inner product $\langle M_a, \alpha_* \rangle$ which is equivalent to the Euclidean inner-product of the vectorized matrices. The Frobenius norm of a matrix is equivalent to the Euclidean norm of the vectorized matrix. The corollary follows as a standard statement characterizing inner products. 
\end{proof}

%% what test statistic
Note that, the production rate is only nonzero if $\alpha_*$ is nonzero. Thus, if \textit{any} linear observable is produced on average, then the process cannot obey detailed balance. This is a special case of the more general observation that, if a process obeys detailed balance, then the expected production rate of all observables is zero. Otherwise, forward and reverse time would be distinguishable. This leads to a natural question. Which observable is the best test statistic?

We will use Lemma \ref{lem: production rates} to seek a linear observable that is an optimal test statistic for detecting violations of detailed balance. We will see that, because the expected production rate of all linear observables is defined by an inner product against the area production rate matrix, $\alpha_*$, the optimal observable is defined in terms of $\alpha_*$ and is produced at a rate which equals a weighted Frobenius norm of $\alpha_*$. 

\vspace{0.1 in}
\noindent \textbf{Given $\alpha_*$, what observable is optimal?}
\vspace{0.05 in}

Suppose that we track a trajectory $\{X(t_k)\}_{k=0}^n$ at a sequence of finely sampled time points $\{t_k\}_{k=0}^n$. Let $\hat{\beta}(\{X(t_k)\}_{k=0}^n)$ denote the empirical approximation to $\beta$ formed by approximating the time-averaged production rate of $B$ over the observed trajectory. In detailed balance, $\beta = 0$ for all observables, so $\hat{\beta} \neq 0$ is evidence against detailed balance. 

The quality of this evidence depends on the variation in $\hat{\beta}$. If $\hat{\beta}$ is highly uncertain, then the confidence interval for $\beta$ given $\hat{\beta}$ may include 0, in which case a hypothesis test would not consider the observed $\hat{\beta}$ sufficient evidence to reject the detailed balance hypothesis. Therefore, $B$ is only a useful observable for detecting violations of detailed balance if the uncertainty in $\beta$ given $\hat{\beta}$ vanishes quickly relative to the expectation of $\hat{\beta}$, which is $\beta$. 

%% define the coefficient of variation
In the absence of the full sampling distribution for $\hat{\beta}$ consider, instead, the coefficient of variation:
\begin{equation}
    \text{CV}(\hat{\beta}) = \frac{\mathbb{V}[\hat{\beta}]^{1/2}}{\mathbb{E}[\hat{\beta}]}
\end{equation}
where $\mathbb{V}[\hat{\beta}]$ is the variance in $\hat{\beta}$. The smaller the coefficient of variation, the less uncertainty in the estimated production rate relative to its expectation. In general, the standard deviation in the numerator will decay at rate $\mathcal{O}(T^{-1/2})$ where $T$ is the length of time observed. The faster the coefficient of variation in $\hat{\beta}$ an observable decays in $T$, the faster the power of the test will grow in $T$. If the coefficient of variation decays sufficiently quickly, then we will not need long trajectories to design a sensitive and specific test.

%% relate to confidence intervals and testing

%% optimality theorem
\begin{Theorem} \label{thm: optimal observable}
    Given an OU process observed at the sequence of samples $t_k = k \Delta t$ for $k = 0$ to $k = n = T/\Delta t$
    \begin{equation}
    \lim_{\Delta t \rightarrow 0, T \rightarrow \infty} T^{1/2} CV(\hat{\beta}(\{X(t_k)\}_{k=0}^n) \geq \|D^{-1/2} \alpha_* \Sigma_*^{-1/2}\|_{\text{Fro}}^{-1} = q^{-1/2}
    \end{equation}
    where $q$ is the entropy production rate for the process. 
    
    The lower bound is tight, and is achieved by any linear observable with vector field proportional to $b(x) = D^{-1} \alpha_* \Sigma_*^{-1} x = D^{-1} v_*(x)$. This observable is produced at rate $\|D^{-1/2} \alpha_* \Sigma_*^{-1/2}\|_{\text{Fro}}^2 = q$, the entropy production rate of the process.
\end{Theorem}

\noindent \textbf{Remark:} The fact that the optimal observable estimates the entropy production rate for the process is natural. Entropy production rates play a fundamental role in the analysis of non-equilibrium systems through fluctuation theorems, which bound the probability of specific trajectories in terms of a large deviation function, which is expressed in terms of entropy production (c.f.~\cite{barato2015thermodynamic,gingrich2016dissipation,speck2004distribution}). Indeed, Theorem \ref{thm: optimal observable} is a special case of the observation, in continuous time Markov jump processes, that the coefficient of variation in any production rate is bounded from below by $q^{-1}$ \cite{barato2015formal,gingrich2016dissipation}. Note that, no physical process can violate thermodynamic equilibrium without exchanging energy between a pair of reservoirs. The entropy production tracks the rate at which it uses energy by tracking the rate at which the trajectory performs work, and in turn, produces or absorbs heat from the reservoir. Thus, it is natural that an estimator for the entropy production should be the optimal test statistic for detecting violations of detailed balance. For other tests based on entropy production rates, see \cite{ciliberto2013heat,ciliberto2013statistical}. For an alternate estimator, see \cite{skinner2021estimating}. For other work relating area and entropy production, see \cite{teitsworth2022stochastic}.

\begin{proof}

Consider a linear observable with vector field $M x + v$. Then by Lemma \ref{lem: production rates}, $\mathbb{E}[\hat{\beta}] = \langle M_a, \alpha \rangle = \langle M, \alpha \rangle$  if $X(0) \sim \pi_*$. 

%% compute the variance for an arbitrary observable
To compute the variance in the estimator $\hat{\beta}(\{X(t_k)\}_{k=0}^n)$ in the joint limit as $T$ diverges and $\Delta t$ approaches zero note that $X(t)$ is ergodic, so long time averages converge in probability to spatial averages. Consider, for example, a procedure in which the average production rate of the observable is approximated by selecting $n$ times steps of length $\Delta t$ from a longer trajectory, with a time interval $\tau$ between sampled steps. Then, if $\tau$ is much larger than the autocorrelation time for the process, $\hat{\beta}(\{X(t_k)\}_{k=0}^n)$ will converge in the limit of large $\tau$ to the sum of $n$ i.i.d.~ random variables of the form $X^{\intercal} M_a^{\intercal} dX$ where $X \sim \pi_*$, and where $dX$ is drawn conditionally given $X$ with a time step of size $\Delta t$. Then, using an ensemble average:
$$
\mathbb{V}[\hat{\beta}(\{X(t_k)_{k=0}^{n})] \rightarrow \mathbb{V} \left[\frac{1}{n \Delta t} \sum_{k=0}^n (M X + v)^{\intercal} dX \right]  = \frac{1}{n \Delta t^2} \mathbb{V}[(M X + v)^{\intercal} dX].
$$
Then, applying the law of total variance:
$$
\begin{aligned}
\mathbb{V}[\hat{\beta}(\{X(t_k)_{k=0}^{n})] & \rightarrow \frac{1}{n \Delta t^2}\left( \mathbb{E}_{X \sim \pi_*}[\mathbb{V}_{dX|X}[(M X + v)^{\intercal} dX]] + \mathbb{V}_{X \sim \pi_*}[\mathbb{E}_{dX|X}[(M X + v)^{\intercal} dX]]\right) \\
& = \frac{1}{n \Delta t^2}\left( \mathbb{E}_{X \sim \pi_*}[(M X + v)^{\intercal} D (M X + v)] \Delta t + \mathbb{V}_{X \sim \pi_*}[(M X + v)^{\intercal} A X] \Delta t^2\right) \\
& = \frac{1}{n \Delta t}\left( \mathbb{E}_{X \sim \pi_*}[(M X + v)^{\intercal} D (M X + v)] + \mathcal{O}(\Delta t)\right) 
\end{aligned}
$$

Therefore, given $T = n \Delta t$:
$$
\begin{aligned}
    \lim_{T \rightarrow \infty} T \mathbb{V}(\hat{\beta}(\{X(t_k)\}_{k=0}^n)) & \rightarrow \mathbb{E}_{X \sim \pi_*}[(M X + v)^{\intercal} D (M X + v)] \\ & = \mathbb{E}_{X \sim \pi_*}[X^{\intercal} M^{\intercal} D M X] + v^{\intercal} D v  = \langle M^{\intercal} D M, \Sigma \rangle + v^{\intercal} D v.
\end{aligned}
$$

To symmetrize the matrix inner product note that, $\langle M^{\intercal} D M, \Sigma \rangle = \text{trace}(M^{\intercal} D M \Sigma) = \text{trace}(\Sigma^{1/2} M^{\intercal} D M \Sigma^{1/2})$ since the trace is preserved by similarity transformations. Then, expanding $D = D^{1/2} D^{1/2}$ gives, $\langle M^{\intercal} D M, \Sigma \rangle  = \| D^{1/2} M \Sigma^{1/2}\|_{\text{Fro}}^2.$ Therefore:
\begin{equation} \label{eqn: asymptotic variance in estimator}
    \lim_{T \rightarrow \infty} T \mathbb{V}(\hat{\beta}(\{X(t_k)\}_{k=0}^n)) = \| D^{1/2} M \Sigma_*^{1/2}\|_{\text{Fro}}^2 + v^{\intercal} D v.
\end{equation}

%% show that we can set v = 0
Introducing $v \neq 0$ does not change the expected production rate $\mathbb{E}[\hat{\beta}]$ but always increases the variance in the estimator since $D$ is positive definite. Therefore, to minimize the coefficient of variation, set $v = 0$. 

Then, to identify an optimal linear observable, find $M \in \mathbb{R}^{d \times d}$ that minimizes:
\begin{equation}
    \lim_{T \rightarrow \infty, \Delta t \rightarrow 0} T^{1/2} \text{CV}(\hat{\beta}) = \frac{\| D^{1/2} M \Sigma_*^{1/2}\|_{\text{Fro}}}{\langle M, \alpha_* \rangle}.
\end{equation}

Equivalently, maximize the ratio:
\begin{equation}
    Z(M) = \frac{\langle M, \alpha_*\rangle}{\| D^{1/2} M \Sigma_*^{1/2}\|_{\text{Fro}}}
\end{equation}
where $Z(M) T^{-1/2}$ is the expected z-score for the observable $\hat{\beta}$ with vector field $M x$ in the long time limit. 

%% show that minimizing the ratio amounts to solving a constrained optimization problem with fixed normalization
Notice that $Z(\lambda M) = Z(M)$ for any $\lambda \in \mathbb{R}$. So without loss of generality, we can restrict our attention to $M$ with a fixed norm. In particular, set $\| D^{1/2} M \Sigma_*^{1/2}\|_{\text{Fro}} = 1$. Then, finding the optimal linear observable amounts to solving the constrained optimization problem:
\begin{equation} \label{eqn: constrained optimization problem}
\begin{aligned}
& \textbf{Find: } M \in \mathbb{R}^d \\
& \textbf{Maximizing: } \langle M, \alpha_* \rangle \\
& \textbf{Given: } \|D^{1/2} M \Sigma_*^{1/2}\|_{\text{Fro}}^2 = 1.
\end{aligned}
\end{equation}
In other words, the optimal observable is given by finding the centered, linear observable with maximal production rate, constrained so that the variance in the observable grows at rate $T$ when $T$ is large. 

%% derive solution
Equation \eqref{eqn: constrained optimization problem} is a constrained optimization problem in a continuous domain with a continuously differentiable equality constraint and a continuously differentiable objective function. Therefore, the solution must satisfy the system of equations:
\begin{equation} \label{eqn: Lagrange}
    \nabla_M\|D^{1/2} M \Sigma_*^{1/2}\|_{\text{Fro}}^2 =  \lambda \nabla_M \langle M, \alpha_* \rangle = \lambda \alpha_*
\end{equation}
for some Lagrange multiplier $\lambda \in \mathbb{R}$. 

The gradient of the weighted Frobenius norm of $M$ is (see Appendix \ref{app: gradient proof}):
\begin{equation} \label{eqn: frobenius norm gradient}
    \nabla_M \|D^{1/2} M \Sigma_*^{1/2}\|_{\text{Fro}}^2 = 2 D M \Sigma
\end{equation}

Therefore, the Lagrange constraint is solved by finding $M$ such that $D M \Sigma \propto \alpha_*$. Then, the optimal observable is given by selecting:
\begin{equation}
    M \propto D^{-1} \alpha_* \Sigma_*^{-1}
\end{equation}

In particular, adopt:
\begin{equation}
    M = D^{-1} \alpha_* \Sigma_*^{-1}, \quad b(x) D^{-1} v_*(x)
\end{equation}
where $v_*(x)$ are the steady state probability velocities (see equation \eqref{eqn: steady state velocities}).

This observable is produced at rate:
$$
\begin{aligned}
    \beta & = \langle M, \alpha \rangle = \text{trace}(D^{-1} \alpha_* \Sigma_*^{-1} \alpha) = \text{trace}(D^{-1/2} \alpha_* \Sigma_*^{-1} \alpha D^{-1/2}) = \|D^{-1/2} \alpha_* \Sigma^{-1/2} \|_{\text{Fro}}^2.
\end{aligned}
$$

%% show equivalence to the entropy production
To conclude, we will show that the rate $\beta = \|D^{-1/2} \alpha_* \Sigma^{-1/2} \|_{\text{Fro}}$ is equivalent to the entropy production rate $q$ for the process. The entropy production rate is the production rate of the observable with vector field $b(x) = -w(x) = D^{-1} A x$ \cite{qian2002thermodynamics,tomita2008irreversible}. 

First, notice that $D^{-1} A = D^{-1} A \Sigma_* \Sigma_*^{-1} = D^{-1} (\frac{1}{2}D + \alpha_*) \Sigma_*^{-1} = \frac{1}{2} \Sigma_*^{-1} + D^{-1} \alpha_* \Sigma_*^{-1}$. By Lemma \ref{lem: production rates}, the entropy production rate only depends on the anti-symmetric part of $D^{-1} A$. Since $\Sigma_*$ is a covariance, $\Sigma_*^{-1}$ is symmetric, so does not influence the entropy production rate. It follows that:
\begin{equation}
    q = \langle D^{-1} \alpha_* \Sigma_*^{-1}, \alpha_* \rangle = \| D^{-1/2} \alpha_* \Sigma_*^{-1/2} \|_{\text{Fro}}^2.
\end{equation}

Therefore the production rate for the observable defined by the vector field $b(x) = D^{-1} \alpha_* \Sigma_*^{-1} x$ is an optimal test statistic for detecting violations of detailed balance, and estimates the entropy production rate for the process. It follows that the expected $z$-score in all linear observables is bounded above, for large $T$ by $q T^{1/2}$ \end{proof}

%% propose estimator and discuss limitations
While elegant, Theorem \ref{thm: optimal observable} is somewhat frustrating, since the optimal observable depends on the diffusion tensor $D$, steady state area production rate matrix $\alpha_*$, and the steady state covariance $\Sigma_*$. These are generally not known in observational problems. In hindsight, this circular dependence is a necessary consequence of the problem structure. For example, it is possible to construct OU processes that only exhibit circulation in one plane. Useful observables should be aligned to track this circulation. Since this plane depends on the process, so does the observable that best illustrates the circulation. 

To use Theorem \ref{thm: optimal observable}, adopt a two-stage process. First, estimate $D$, $\alpha_*$, and $\Sigma_*$. These estimates could be formed by tracking an initial trajectory segment, by running a bootstrap procedure on increments extracted from a full trajectory. Alternately, use an initial trajectory segment to form a maximum likelihood estimator for $D$ and $A$, or for $D$ and $\alpha_* \Sigma_*^{-1}$. Second, use the estimated matrices ato form $b(x)$ and to estimate $q$ via a time averaged path-sum. Ideally, the data should be divided so that the information used to form $D^{-1} \alpha_* \Sigma_*^{-1}$ is independent of the information later used to estimate $q$ from $\hat{\beta}$. 

The optimality guarantee developed in Theorem \ref{thm: optimal observable} also relies on the convergence of long-time averages to spatial averages against the steady-state. Therefore, the optimality guarantee is only provides relevant advice for trajectories that are watched for a long time, that is, for long enough for the trajectory to mix. Given a long trajectory, demonstrating violations of detailed balance is only difficult if the violations are small. Therefore, Theorem \ref{thm: optimal observable} provides relevant advice when the user is attempting to detect small violations of detailed balance using long trajectories. In this setting, dividing the data into two separate sets in order to perform the two-stage estimation procedure described above is not too wasteful.

%%%%%%%%%%%%%%%%%%%%%%%%%%%%%%%%%%%%%%%%%%
\section{Conclusion} 

%% summary paragraph
In summary, the area production rate matrix introduced in \cite{ghanta2017fluctuation} and \cite{gonzalez2019experimental} characterizes the nonequilibrium behavior of OU processes. Following \cite{tomita1974irreversible} and \cite{tomita2008irreversible} we showed that the probability fluxes and velocities are generated by the instantaneous area production rate matrix $\alpha(t)$, which coincides with the angular momentum matrix of the probability fluid. At steady state, it converges to $\alpha_*$, which is the unique linear mapping that converts the conservative component of the vector field $w(x)$ into its rotational component, where $w(x)$ is defined such that path integration against $w(x)$ returns the work performed along a trajectory, and where minimizing the work along a path is equivalent to solving for a maximally likely path \cite{nolting2016balls}. We then showed that, at least in two-dimensions, the steady-state distribution of a nonequilibrium OU process depends on induced rotation in a geometric fashion that is elegantly expressed using the corresponding area production rate (see Theorem \ref{thm: steady state ellipse geometry}). To conclude, we illustrated that the expected long-term production rate of any observable defined as a path integral against a linear vector field is controlled by an inner product against $\alpha_*$. As a result, the fastest any linear observable is produced is $\|\alpha_*\|_{\text{Fro}}^2$. This result recommends the norm area production rate matrix as a signature of nonequilibrium behavior since, at equilibrium, no observable is produced on average. In particular, the linear observable with the asymptotically smallest coefficient of variation in its empirical estimate returns a weighted Frobenius norm of the area production rate matrix that coincides with the entropy production rate of the process (see Theorem \ref{thm: optimal observable}). 

%%%%%%%%%%%%%%%%%%%%%%%%%%%%%%%%%%%%%%%%%%
\bibliographystyle{siam}
\bibliography{Refs.bib}

%%%%%%%%%%%%%%%%%%%%%%%%%%%%%%%%%%%%%%%%%%
\pagebreak
\section{Appendix}

%%%%%%
\subsection{Proof of Equation \eqref{eqn: trace integral}} \label{app: trace proof}

Consider an integral of the form (where $M$ is an arbitrary square matrix and $\mu(x)$ is differentiable distribution that vanishes at least exponentially at infinity):
\begin{equation}
    -\int_{\mathbb{R}^d} x^T M \nabla \mu(x) dx.
\end{equation}
    
Now, notice that $\nabla \cdot [(M^T x) \mu(x)] = [\nabla \cdot M^T x] \mu(x) + (M^T x) \cdot \nabla \mu(x)$. Therefore $x^T M \nabla \mu(x) = \nabla \cdot [(M^T x) \mu(x)] - [\nabla \cdot M^T x] \mu(x).$ Now, the divergence of $M^T x$ is just the trace of $M^T$, which is equal to the trace of $M$. Therefore the integral can be rewritten:
\begin{equation}
-\int_{\mathbb{R}^d} x^T M \nabla \mu(x) dx = \int_{\mathbb{R}^d} \text{trace}(M) \mu(x) dx - \int_{\mathbb{R}^d} \nabla \cdot [(M^T x) \mu(x)].
\end{equation}
    
To evaluate the second integral use the divergence theorem. If we evaluate the integral over some domain $\Omega$:
\begin{equation}
    -\int_{\Omega} \nabla \cdot [(M^T x) \mu(x)] \mu(x) dx = \int_{\partial \Omega} \partial_n [(M^T x) \mu(x)] dx.
\end{equation}
    
Now, taking the domain to be a sphere radius $r$, all partial derivatives of $(M^T x \mu(x))$ are vanishing for large $r$ since $\mu(x)$ and its partials are assumed to vanish exponentially at infinity. Since the surface area of the sphere only expands at rate $\mathcal{O}(r^{d - 1})$ the exponential decay of $\mu$ ensures that the first integral vanishes. Then:
\begin{equation}
    -\int_{\mathbb{R}^d} x^T M \nabla \mu(x) dx = \text{trace}(M) \int_{\mathbb{R}^d} \mu(x) dx = \text{trace}(M).
\end{equation}
Therefore any integral of this form simply reduces to the trace of the matrix $M$. $\square$

%%%%%%
\subsection{Proof of Theorem \ref{thm: steady state ellipse geometry} and Corollaries \ref{corr: bounding distribution} through \ref{corr: form}} \label{app: steady state proofs}

Consider a two dimensional OU process with the standard parameterization:
\begin{equation}
dZ  = \bar{\lambda} \left[\begin{array}{cc} (1 + \mu) & 2 \omega\\ -2 \omega & (1 - \mu) \end{array} \right] Z dt + dW = A(\omega) Z dt + dw
\end{equation}

Then let $\Sigma_*(\omega)$ denote the steady state covariance and $\Sigma_*(\omega)^{-1}$ its inverse as functions of the area production rate $\omega$. Then, $\Sigma_*(\omega)$ must solve the Lyapunov equation:
\begin{equation} \label{eqn: appendix lyapunov}
    A(\omega) \Sigma_*(\omega) + \Sigma_*(\omega) A(\omega)^{\intercal} = I.
\end{equation}

We will parameterize $\Sigma_*(\omega)$ by parameterizing the singular value decomposition of its inverse. In particular, we will show that $\Sigma_*(\omega)^{-1}$ has singular values:
\begin{equation}
    s_\pm(\omega) = \bar{s}(\omega)(1 \pm s'(\omega))
\end{equation}
where:
\begin{equation} \label{eqn: singular values of precision}
\begin{aligned}
    & \bar{s}(\omega) = \frac{1}{2}\text{trace}(\Sigma_*(\omega)^{-1}) = \text{trace}(A(\omega)) = 2 \bar{\lambda} , \quad s'(\omega) = \sqrt{\frac{\mu^2}{1 + (2\omega)^2}}
\end{aligned}
\end{equation}
and has singular vectors rotated $\theta(\omega)$ degrees off the coordinate bases where:
\begin{equation} \label{eqn: appendix angle}
    \theta(\omega) = \frac{1}{2} \tan^{-1}(2 \omega).
\end{equation}

If shown, these results would immediately establish Corollaries \ref{corr: SVD} and \ref{corr: form}, as well as the first statement in Corollary \ref{corr: monotonicity}. The remaining statements in Corollaries \ref{corr: monotonicity} and \ref{corr: entropy} follow by direct analysis of the singular values. Corollary \ref{corr: bounding distribution} follows by explicit limiting analysis of the ensuing explicit form for the steady state covariance. 

Therefore, it remains to show that: (a) if $\Sigma_*^{-1}(\omega)$ is parameterized by equations \eqref{eqn: singular values of precision} and \eqref{eqn: appendix angle}, then $\Sigma_*(\omega)$ solves the Lyapunov equation and that $\Sigma_*(\omega)$ generates a family of ellipses and density functions satisfying the bounds stated in Theorem \ref{thm: steady state ellipse geometry} and Corollary \ref{corr: bounding distribution}.

First, we establish that $\bar{s}(\omega) = 2 \bar{\lambda}$ for all $\omega$. Multiply the Lyapunov equation \eqref{eqn: appendix lyapunov} from the left by $\Sigma_*^{-1}$, then compute the trace:
$$
    \text{trace}(\Sigma_*^{-1} A \Sigma_*) + \text{trace}(A^{\intercal}) = 2 \text{trace}(A) = \text{trace}(\Sigma_*^{-1}).
$$

Therefore, $\text{trace}(\Sigma_*^{-1}) = 2 \bar{s}(\omega) = 2 \text{trace}(A) = 4 \bar{\lambda}$ so $\bar{s}(\omega) = 2 \bar{\lambda}$.  
	
Given $\theta(\omega)$, $\Sigma_*(\omega)$ and its inverse have singular vectors:
$$
U(\omega) = \left[\begin{array}{cc} \cos(\theta(\omega)) & \sin(\theta(\omega)) \\ -\sin(\theta(\omega)) & \cos(\theta(\omega)) \end{array} \right]
$$

The diagonal matrix of singular values is:
$$
S^{-1}(\omega) = \frac{1}{\bar{s}}\left[\begin{array}{cc} \frac{1}{1 + s'(\omega)} & 0 \\ 0 & \frac{1}{1 + s'(\omega)} \end{array} \right] = \frac{1}{2 \bar{\lambda}} \left[\begin{array}{cc} \frac{1}{1 + s'(\omega)} & 0 \\ 0 & \frac{1}{1 + s'(\omega)} \end{array} \right]
$$
	
Then $\Sigma_*(\omega) = U(\omega) S^{-1}(\omega) U^{\intercal}(\omega)$ has entries:
$$
\Sigma_*(\omega) = \left[\begin{array}{cc} \frac{u_{11}^2}{s_+} + \frac{u_{12}^2}{s_-}  & \frac{u_{11} u_{12}}{s_+} + \frac{u_{12} u_{22}}{s_-} \\
\frac{u_{21} u_{11}}{s_+} + \frac{u_{12} u_{22}}{s_-}  & \frac{u_{21}^2}{s_+} + \frac{u_{22}^2}{s_-} \end{array} \right].
$$
where $s_\pm = \bar{s} (1 \pm s'(\omega))$. 

Next, simplify the products of the entries of $U$:
$$
\begin{aligned}
& u_{11} = u_{22} = \cos{\left(\frac{1}{2} \tan^{-1}(\omega) \right)} = \sqrt{\frac{1}{2} \left(1 + \frac{1}{\mu} s'(\omega)\right)} \\
& u_{12} = - u_{21} = \sin{\left(\frac{1}{2} \tan^{-1}(\omega) \right)} = \text{sign}(\omega) \sqrt{\frac{1}{2} \left(1 - \frac{1}{\mu} s'(\omega) \right)}.
\end{aligned}
$$
Therefore:
$$
\begin{aligned}
& u_{11}^2 = u_{22}^2 = \frac{1}{2} \left(1 + \frac{1}{\mu} s'(\omega) \right) \\
& u_{12}^2 = u_{21}^2 = \frac{1}{2} \left(1 - \frac{1}{\mu} s'(\omega) \right) \\
& u_{11} u_{12} = -u_{21} u_{22} = \frac{1}{2} \text{sign}(\omega) \sqrt{1 - \left(\frac{s'(\omega)}{\mu} \right)^2}.
\end{aligned}
$$
	
The covariance matrix is symmetric so we only need to evaluate the diagonal entries and one of the off-diagonals. The first diagonal entry is:
$$
\begin{aligned}
\frac{u_{11}^2}{s_+} + \frac{u_{12}^2}{s_-}  & = \frac{1}{4 \lambda} \left( \frac{1 + \frac{1}{\mu} s'(\omega)}{1 + s'(\omega)} + \frac{1 - \frac{1}{\mu} s'(\omega)}{1 - s'(\omega)} \right) \\
& = \frac{1}{4\lambda} \left(\frac{1}{1 + s'(\omega)} + \frac{1}{1 - s'(\omega)} + \frac{s'(\omega)}{\mu}\left(\frac{ 1}{1 + s'(\omega)} - \frac{1}{1 - s'(\omega)} \right) \right) \\
& = \frac{1}{4\lambda} \left(\frac{2}{1 - s'(\omega)^2} + \frac{s'(\omega)}{\mu}\left(\frac{ 2 s'(\omega)}{1 - s'(\omega)^2} \right) \right) \\
& = \frac{1}{2 \lambda} \frac{1}{1 - s'(\omega)^2} \left(1 + \frac{s'(\omega)^2}{\mu} \right)
\end{aligned}
$$
	
The second diagonal entry simplifies similarly:
$$
\begin{aligned}
\frac{u_{21}^2}{s_+} + \frac{u_{22}^2}{s_-}   = \frac{1}{2 \lambda} \frac{1}{1 - s'(\omega)^2} \left(1 - \frac{s'(\omega)^2}{\mu} \right).
\end{aligned}
$$

The off-diagonal entries are:
$$
\begin{aligned}
\frac{u_{21} u_{11}}{s_+} + \frac{u_{12} u_{22}}{s_-}  & = \frac{1}{4 \lambda} \left( \frac{-\text{sign}(\omega)}{1+s'(\omega)} \sqrt{1 - \left(\frac{s'(\omega)}{\mu} \right)^2} \right) + \frac{1}{4 \lambda} \left( \frac{\text{sign}(\omega)}{1-s'(\omega)} \sqrt{1 - \left(\frac{s'(\omega)}{\mu} \right)^2} \right) \\ & = \frac{1}{4 \lambda} \text{sign}(\omega) \sqrt{1 - \left(\frac{s'(\omega)}{\mu} \right)^2} \left(\frac{1}{1-s'(\omega)} - \frac{1}{1+s'(\omega)} \right) \\
& = \frac{\text{sign}(\omega)}{2 \lambda} \frac{s'(\omega)}{1 - s'(\omega^2)} \sqrt{1 - \left(\frac{s'(\omega)}{\mu} \right)^2}.
\end{aligned}
$$

Therefore:
$$
\Sigma_*(\omega) = \frac{1}{2 \bar{\lambda}} \frac{1}{1 - s'(\omega)^2} \left[I + s'(\omega) \left[\begin{array}{cc} s'(\omega)/\mu & \text{sign}(\omega) \sqrt{1 - \left(\frac{s'(\omega)}{\mu} \right)^2} \\ \text{sign}(\omega) \sqrt{1 - \left(\frac{s'(\omega)}{\mu} \right)^2} & -s'(\omega)/\mu \end{array}\right] \right]
$$
	
Next, use \eqref{eqn: singular values of precision} to substitute in for $s'(\omega)$. This gives:
\begin{equation}
\Sigma_*(\omega) = \frac{1}{2 \lambda} \left(1 - \frac{\mu^2}{1 + \left(2 \omega \right)^2}\right)^{-1} \left[I + \frac{\mu}{1 + \left(2 \omega \right)^2} \left[\begin{array}{cc} 1 & 2 \omega \\ 2 \omega & -1 \end{array}\right] \right]
\end{equation}
	
Therefore, Corollaries \ref{corr: SVD} and \ref{corr: form} define the same matrix-valued function of $\omega$. It remains to show that this matrix solves the Lyapunov equation, and thus recovers the steady state covariance. To show that $\Sigma_*(\omega)$ satisfies the Lyapunov equation \eqref{eqn: appendix lyapunov}, we compute the products $A(\omega) \Sigma_*(\omega)$ and $\Sigma_*(\omega) A(\omega)^{\intercal}$ explicitly.
$$
\begin{aligned}
& A(\omega) \Sigma_*(\omega) = \frac{1}{2} \left(1 - \frac{\mu^2}{1 + (2 \omega)^2} \right)^{-1} \left[\left[\begin{array}{cc} 1 - \mu & (2 \omega) \\ -(2 \omega) & 1 + \mu \end{array} \right] + \frac{\mu}{1+(2 \omega)^2} \left[\begin{array}{cc} 1- \mu + (2 \omega)^2 & -\mu (2 \omega) \\ \mu (2 \omega) & -(1 + \mu) - (2 \omega)^2 \end{array} \right] \right].\\
& \Sigma_*(\omega) {A}^T(\omega)  = \frac{1}{2} \left(1 - \frac{\mu^2}{1 + (2 \omega)^2} \right)^{-1} \left[\left[\begin{array}{cc} 1 - \mu & -(2 \omega) \\ (2 \omega) & 1 + \mu \end{array} \right] + \frac{\mu}{1+(2 \omega)^2} \left[\begin{array}{cc} 1- \mu + (2 \omega)^2 & \mu (2 \omega) \\ -\mu (2 \omega) & -(1 + \mu) - (2 \omega)^2 \end{array} \right] \right].
\end{aligned}
$$
	
Notice that the off-diagonal terms in the second product are equal to the negative of the off-diagonal terms in the first product. Therefore, summing the two gives a diagonal matrix. All that remains is to show that this matrix is the identity. Taking the sum:
$$
\begin{aligned}
A(\omega) \Sigma_*(\omega) + \Sigma_*(\omega) A^{\intercal}(\omega) = & \left(1 - \frac{\mu^2}{1 + (2 \omega)^2} \right)^{-1} \times \ldots \\ & \ldots \left[\begin{array}{cc} (1 - \mu) + \frac{\mu}{1 + (2 \omega)^2}((1 - \mu) + (2 \omega)^2)  & 0 \\ 0 & (1 + \mu) - \frac{\mu}{1 + (2 \omega)^2}((1 + \mu) + (2 \omega)^2)  \end{array} \right].
\end{aligned}
$$
	
To simplify, multiply by $(1 + (2 \omega)^2)/(1+(2 \omega)^2)$. This gives:
$$
A^(\omega) \Sigma_*(\omega) + \Sigma_*(\omega) A^{\intercal}(\omega) = \left(1 + (2 \omega)^2 - \mu^2 \right)^{-1}  \left[\begin{array}{cc} 1 + (2 \omega)^2 - \mu^2  & 0 \\ 0 & 1 + (2 \omega)^2 - \mu^2  \end{array} \right] = I.
$$
	
It follows that $\Sigma_*(\omega)$ satisfies the Lyapunov equation for all $\omega$. Thus, Corollaries \ref{corr: SVD} and \ref{corr: form} correctly parameterize the steady state covariance. 

To conclude, we show that, if $\Sigma_*(\omega)$ satisfies Corollary \ref{corr: SVD} then the bounding statements in Theorem \ref{thm: steady state ellipse geometry} and Corollary \ref{corr: bounding distribution} hold. 

%% prove bounding ellipses
	
The next key geometric observation that we need to prove is to prove that the ellipses $\mathcal{E}(\Sigma_*^{-1}(\omega))$ defined by $\frac{1}{2} z^T \Sigma^{-1}(\omega) z = 1$ is bounded by the ellipses $\mathcal{E}(\Lambda^-)$ and $\mathcal{E}(\Lambda^+)$, where the first is given by $\bar{\lambda} z_1^2 + \lambda_2 z_2^2 = 1$ and the second by $\lambda_1 z_1^2 + \bar{\lambda} z_2^2 = 1$, for all $\omega$, and is always tangent to the two ellipses at least two points.
	
Define the potential functions:
\begin{equation}
\begin{aligned}
& u^+(z) = z^T \Lambda^+ z,\quad \Lambda^+ = \text{diag}(\lambda_+, \bar{\lambda}) \\
& u_*(z,\omega) = \frac{1}{2} z^T \Sigma_*^{-1}(\omega) z \\
& u^-(z) = z^T \Lambda^{-} z,\quad \Lambda^- = \text{diag}(\bar{\lambda}, \lambda_-).
\end{aligned}
\end{equation}
The relevant ellipses are level sets of these potential functions when they are set equal to one. 

To show that the iso-ellipse of the potential function $u(z,\omega)$ is contained inside the iso-ellipse of the potential function $u^{-}(z)$ for all $\omega$, and tangent to it at two points, we will minimize the potential function $u_*(z,\omega)$ constrained to the ellipse $u^{-}(z) = 1$. If the minimum value of $u_*(z,\omega)$ acheived is equal to one for all $\omega$ then the ellipse $u_*(z,\omega) = 1$ is contained inside the ellipse $u^{-}(z) = 1$ for all $\omega$ and is tangent to it at at least two points. To show that the iso-ellipse of the function $u_*(z,\omega)$ contains the iso-ellipse of $u^{+}(z)$ for all $\epsilon$ we maximize the value of $u_*(z,\omega)$ constrained to the ellipse $u^{+}(z) = 1$. Then, if the maximum value acheived is always one, the iso-ellipse $u_*(z,\omega) = 1$ contains the iso-ellipse $u^{+}(z) = 1$ and is tangent to it at at least two points. 
	
 Start with the outer ellipse. Minimizing $u_*(z,\omega)$ constrained to $u^{-}(z) = 1$ is a constrained minimization problem. Stationary requires:
\begin{equation}
\nabla u^{-}(z) = \nu \nabla u_*(z,\omega) 
\end{equation}
for some Lagrange multiplier $\nu$. Evaluating the gradients yeilds:
$$
\Lambda^- z = \nu \Sigma_*^{-1}(\omega) z.
$$
	
Multiply both sides with $\Sigma_*(\omega)$ to convert to an eigenvalue problem:
\begin{equation}
\Sigma_*(\omega) \Lambda^- z = \nu z.
\end{equation}
	
Therefore, any minimizing $z$ must be parallel to an eigenvector of $\Sigma_*(\omega) \Lambda^-$, where $\nu$ is the corresponding eigenvalue.
	
Now, the constraint required $z^T \Lambda^- z = 1$. But, $\Lambda^- z = \nu \Sigma_*^{-1}(\omega) z$ so:
$$
z^T \Lambda^- z = \frac{2 \nu}{2} z^T \Sigma_*^{-1}(\omega) z = 2 \nu u_*(z,\omega) = 1. 
$$
	
It follows that, if $z$ is an eigenvector of $\Sigma_*(\omega) \Lambda^-$ with eigenvalue $\nu$, then the value of the potential function $u_*(z,\omega)$ at $z$ is:
\begin{equation}
u_*(z,\omega) = \frac{1}{2 \nu}.
\end{equation}
	
So, to minimize the potential function $u_*(z,\omega)$, maximize $\nu$. Therefore, we are interested in is the eigenvector of $\Sigma_*(\omega) \Lambda^-$ corresponding to the largest eigenvalue of $\Sigma_*(\omega) \Lambda^-$. Repeating this analysis for the inner ellipse shows that the intercept of interest with the inner ellipse occurs along the eigenvalue of $\Sigma_*(\omega) \Lambda^+$ corresponding to the smallest eigenvalue of $\Sigma_*(\omega) \Lambda^+$. 
This reduces the problem of finding the intersects to an eigenvalue problem. In each case the eigenvector of interest corresponds to either the largest or smallest eigenvalue. We aim to show that the value of $u_*(z,\omega)$ is equal to one at both of these points, regardless of $\omega$. So, we must show that the largest eigenvalue of $\Sigma_*(\omega) \Lambda^-$ and the smallest eigenvalue of $\Sigma_*(\omega) \Lambda^+$ are both always equal to $1/2$ for any $\omega$.
	
All of the matrices involved are two dimensional so we can compute their eigenvalues explicitly by finding their trace and determinant. Specifically, given a two by two matrix with trace $t$ and determinant $d$ the eigenvalues are:
$$
\nu_{\pm} = \frac{t \pm \sqrt{t^2 - 4 d}}{2}.
$$
	
Therefore, to find the eigenvalues of $\Sigma_*(\omega) \Lambda^-$ and $\Sigma_*(\omega) \Lambda^+$ we will compute their trace and determinant. The determinant of $\Sigma_*(\omega)$ is:
\begin{equation}
|\Sigma_*(\omega)| = \frac{1}{4 \lambda^2} \frac{1}{1 - s'(\omega)^2}.
\end{equation}
	
The determinant of the product of matrices is the product of the determinants. Both $\Lambda^-$ and $\Lambda^+$ are diagonal so their determinants are just the products of their entries. Then:
$$
|\Sigma_*(\omega) \Lambda^-| = \frac{1}{4} \frac{1 - \mu}{1 - s'(\omega)^2} \quad |\Sigma_*(\omega) \Lambda^+| = \frac{1}{4} \frac{1 + \mu}{1 - s'(\omega)^2}.
$$
	
To find the traces we compute the diagonal entries of the products explicitly. For example:
\begin{equation}
\text{trace}(\Sigma_*(\omega) \Lambda^-) = \frac{1}{2} \left[1 + \frac{1 - \mu}{1 - s'(\omega)^2} \right] = \frac{1}{2} \left[1 + 4 |\Sigma_*(\omega) \Lambda^-| \right].
\end{equation}
	
Let $t$ be the trace and $d$ be the determinant. Then we have shown that, for $\Sigma_*(\omega) \Lambda^-$, $t = \frac{1}{2} + 2 d$. Plugging into the quadratic equation recovers the eigenvalues:
\begin{equation}
\begin{aligned}
\nu_{\pm} & = \frac{t \pm \sqrt{t^2 - 4 d}}{2} = \frac{\frac{1}{2} + 2 d}{2} \pm \frac{\sqrt{\frac{1}{4} + 2 d + 4 d^2 - 4 d}}{2} \\ & = \frac{1}{4} +  d \pm \frac{1}{2}\sqrt{(1/2 - 2 d)^2} = \frac{1}{4} + d \pm \left(\frac{1}{4} - d \right)
\end{aligned} 
\end{equation}
	
We are looking for the larger eigenvalue, which is given by summing the two terms, which gives $\beta_+ = \frac{1}{4} + d + \frac{1}{4} - d = \frac{1}{2}.$ Therefore, the largest eigenvalue of $\Sigma_*(\omega) \Lambda^-$ is always equal to one half, so the minimum value of $u_*(z,\omega)$ constrained to $u^{-}(z) = 1$ is equal to one. This proves that the ellipse $\mathcal{E}(\Sigma_*(\omega))$ given by setting $u_*(z,\omega) = 1$ is always contained inside the ellipse $u^{-}(z) = 1$, and tangent to it at at least two points. Moreover, these points are parallel to the eigenvector of $\Sigma_*(\omega) \Lambda^-$ corresponding to eigenvalue $1/2$. 
	
Repeating for the inner ellipse:
\begin{equation}
 \begin{aligned}
& |\Sigma_*(\omega) \Lambda^+| = \frac{1}{4} \frac{1+\mu}{1 - s'(\omega)^2} \\
& \text{trace}(\Sigma_*(\omega) \Lambda^-) = \frac{1}{2} \left[ 1 - \frac{1 + \mu}{1 - s'(\omega)^2} \right] = \frac{1}{2} - 2 |\Sigma_*(\omega) \Lambda^+|. 
 \end{aligned}
\end{equation}
	
Therefore, when considering the inner ellipse the trace is $t = \frac{1}{2} - 2 d$, where $d$ is the determinant. Plugging this into the quadratic equation confirms that the smallest eigenvalue is always $1/2$, so the maximal value of $u_*(z,\omega)$ constrained to $u^{+}(z) = 1$ is one. This proves that the ellipse $\mathcal{E}(\Sigma_*(\omega))$ given by fixing $u_*(z,\omega) = 1$ always contains the ellipse fixed by $u^{+}(z) = 1$, and is tangent to it at at least two points. Now the eigenvector of $\Sigma_*(\omega) \Lambda^+$ corresponding to eigenvalue $1/2$ points to the intersection points. This could lead to a formula for finding the angle between the vector pointing from the origin to the intersection of the two ellipses and the $y$ axis.
	
This proves that, for all $\omega$ the ellipse defined by setting $u_*(z,\omega) = 1$ is bounded between the ellipses $u^{-}(z) = 1$ and $u^{+}(z) = 1$, and tangent to both at least two points. As long as $\lambda_1 \neq \lambda_2$ the outer and inner ellipse are distinct, so the ellipse $u_*(z,\omega) = 1$ is tangent to both at exactly two points. Since the ellipses always intersect there are no other ellipses contained in the ellipse defined by $u^{-}(z) = 1$, or containing the ellipse $u^{+}(z) = 1$, such that the ellipse $u_*(z,\omega) = 1$ is always contained in, or contains the intermediate ellipse for all $\epsilon$. This proves that the bounding ellipses $u^{-}(z) = 1$ and $u^{+}(z) = 1$ are the smallest, and largest ellipses that contain, or are contained by, $u_*(z,\omega)$ for all $\omega$. As an immediate consequence, replacing the covariance with $\Lambda^-$ or $\Lambda^+$ automatically produces an upper bound and a lower bound on the steady-state distribution. These are the tightest such bounds achievable using any distribution that is independent of $\omega$ (Corollary \ref{corr: bounding distribution}).

It remains, only, to show that the four points where the ellipse $u_*(z,\omega) = 1$ intersects the bounding ellipses are orthogonal points rotated $2 \theta(\omega)$ degrees off the coordinate axis. the intersection points occur along the eigenvectors of $\Sigma_*(\omega) \Lambda^-$ and $\Sigma_*(\omega) \Lambda^+$ that correspond to eigenvalues $1/2$. Since the eigenvalues are known we only need to find the null-vectors of the matrices $\Sigma_*(\omega) \Lambda^- - \frac{1}{2} I$ and $\Sigma_*(\omega) \Lambda^+ - \frac{1}{2} I$. The first matrix is:
$$
\Sigma_*(\omega) \Lambda^+ - \frac{1}{2} I = \Sigma_*(\omega) \Lambda^+ - \frac{1}{2} I = \frac{1}{2} (1 - s'(\omega))^2 \frac{s'(\omega)^2}{\mu}\left[\begin{array}{cc} 1 - \frac{\mu^2}{s'(\omega)^2} & (2 \omega) \\ (2 \omega) (1 - \mu) & -(1 - \mu) \end{array} \right].
$$
	
But $s'(\omega) = \mu (1 + (2 \omega)^2)^{-1/2}$ so the term $\mu^2/s'(\omega)^2$ is just $1 + (2 \omega)^2$ therefore:
$$
\Sigma_*(\omega) \Lambda^+ - \frac{1}{2} I = \frac{1}{2} (1 - s'(\omega))^2 \frac{s'(\omega)^2}{\mu}\left[\begin{array}{cc} -(2 \omega)^2 & (2 \omega) \\ (2 \omega) (1 - \mu) & -(1 - \mu) \end{array} \right].
$$
	
This matrix has null vector:
$$
v = \frac{1}{\sqrt{1 + (2 \omega)^2}} \left[\begin{array}{c} 1 \\ (2 \omega) \end{array} \right]
$$
Since if $v_2 = (2 \omega) v_1$ then $-(2 \omega)^2 v_1 + (2 \omega) v_2 = 0$ and $(1 - \mu)((2 \omega) v_1 + v_2) = 0$. It follows that the vector pointing to the intersection with the outer ellipse is parallel to $[1, 2 \omega]$ where $(2 \omega)$ is the area production. When $(2 \omega)$ is zero (detailed balance) this vector is just $[1,0]$ so is parallel with the first coordinate axes. When $(2 \omega)$ goes to positive infinity it converges to $[0,1]$ and when $(2 \omega)$ goes to negative infinity it goes to $[0,-1]$ so the angle measured with the first coordinate axis approaches $\pi/2$ and $-\pi/2$ as rotation becomes infinitely positive or negative. The angle is given by:
\begin{equation}
\theta_{outer}(\omega) = \cos^{-1}\left(\frac{1}{\sqrt{1 + (2 \omega)^2}} \right) = \tan^{-1}((2 \omega)) = 2 \theta(\omega).
\end{equation}
This proves that the vector pointing to the intersection with the outer ellipse is given by rotating off the first coordinate axes by twice the tilt angle. $\square$

%%%%%%
\subsection{Proof of Equation \eqref{eqn: frobenius norm gradient}} \label{app: gradient proof}

Consider the gradient $\nabla_M \|A^{\intercal} M B\|_{\text{Fro}}^2$ for some pair of matrices $A$ and $B$. To compute the gradient, first expand the norm explicity:
$$
\begin{aligned}
\|A^{\intercal} M B\|_{\text{Fro}}^2 & = \sum_{ij} [A^{\intercal} M B]_{ij}^2 = \sum_{i,j} \sum_{k,l} \sum_{i,j} \left( \sum_{k,l} a_{ik} m_{kl} b_{jl} \right)^2 \\
& = \sum_{i,j} \sum_{k,l} \sum_{m,n} a_{ik} a_{in} m_{kl} m_{nm} b_{jl} b_{jm} = \sum_{k,l,m,n} \left[\sum_{ij} a_{ik} a_{in} \right] m_{kl} m_{nm} \left[\sum_{ij} b_{jl} a_{jm} \right] \\
& = \sum_{k,l,n,m} \left[A^{\intercal} A \right]_{kn} m_{kl} m_{nm} [B^{\intercal} B]_{lm}.
\end{aligned}
$$

Then, differentiating with respect to $m_{ij}$ gives:
$$
\partial_{m_{ij}} \|A^{\intercal} M B\|_{\text{Fro}}^2 = \sum_{kl} [A^{\intercal} A]_{ki} m_{kl} [B^{\intercal} B]_{lj} + \sum_{n,m} [A^{\intercal} A]_{in} m_{nm} [B^{\intercal} B]_{jm} 
$$
the two sums are identical since $A^{\intercal} A$ and $B^{\intercal} B$ are both symmetric matrices. Therefore:
$$
\partial_{m_{ij}} \|A^{\intercal} M B\|_{\text{Fro}}^2 = \sum_{kl} [A^{\intercal} A]_{ki} m_{kl} [B^{\intercal} B]_{lj} = \left[ (A^{\intercal} A)^{\intercal} M (B^{\intercal} B)  \right]_{ij}
$$
so:
\begin{equation}
    \nabla_M \|A^{\intercal} M B\|_{\text{Fro}}^2 = (A^{\intercal} A) M (B^{\intercal} B).
\end{equation}

Substituting $D^{1/2}$ for $A$ and $\Sigma^{1/2}$ for $B$ recovers the stated result. $\square$

\end{document}